\documentclass[12pt,oneside,english]{amsart}

\usepackage[T1]{fontenc}
\usepackage[latin9]{inputenc}
\usepackage{geometry}
\usepackage{babel}
\usepackage{array}
\usepackage{multirow}
\usepackage{amsthm}
\usepackage{amsbsy}
\usepackage{amssymb}
\usepackage{graphicx}

\makeatletter

\numberwithin{equation}{section}
\numberwithin{figure}{section}
\theoremstyle{plain}
\newtheorem{thm}{\protect\theoremname}[section]
  \theoremstyle{definition}
  \newtheorem{defn}[thm]{\protect\definitionname}
  \theoremstyle{plain}
  \newtheorem{lem}[thm]{\protect\lemmaname}
  \theoremstyle{remark}
  \newtheorem{rem}[thm]{\protect\remarkname}
  \theoremstyle{definition}
  \newtheorem{example}[thm]{\protect\examplename}
  \theoremstyle{plain}
  \newtheorem{conjecture}[thm]{\protect\conjecturename}

\makeatother

  \providecommand{\conjecturename}{Conjecture}
  \providecommand{\definitionname}{Definition}
  \providecommand{\examplename}{Example}
  \providecommand{\lemmaname}{Lemma}
  \providecommand{\remarkname}{Remark}
\providecommand{\theoremname}{Theorem}

\begin{document}

\title{K-Theory and Pseudospectra for Topological Insulators}

\author{Terry A. Loring}

\address{Department of Mathematics and Statistics, University of New Mexico,
Albuquerque, NM 87131, USA.}
\begin{abstract}
We derive formulas and algorithms for Kitaev's invariants in the
periodic table for topological insulators and superconductors for
finite disordered systems on lattices with boundaries. We find
that $K$-theory arises as an obstruction to perturbing approximately compatible
observables into compatible observables. 

We derive formulas in all
symmetry classes up to dimension two, and in one symmetry class in
dimension three, that can be computed with sparse matrix algorithms.
We present algorithms in two symmetry classes in 2D and one in 3D
and provide illustrative studies regarding how these algorithms can
detect the scaling properties of phase transitions. 
\end{abstract}
\maketitle

\section{Approximately compatible observables\label{sec:Introduction}}

Compatible observables are given by a rigid definition. If they act
on finite Hilbert space, the requirement is a basis of vectors that
are completely localized for each $X_{j}$, so 
$X_{j}\mathbf{v}=\lambda_{j}\mathbf{v}$
for some scalars. If we repeatedly prepare precisely the same state,
we can hope to get no variance in whichever $X_{j}$ we then measure.
This sounds more like a math theory than any laboratory. What if we
repeatedly prepare approximately the same state and approximately
measure one of the $X_{j}$ in turn, and each series of measurements
shows little variation? Might this happen because the $X_{j}$ are
\emph{approximately compatible?}

In the end, we expect this to mean the commutators 
$\left[X_{j},X_{k}\right]=X_{j}X_{k}-X_{k}X_{j}$
are small, preferably in the operator norm. However the initial definition
should involve something like small variance of states. More simply,
we can seek common approximate eigenvectors. A fundamental quantity
would seem to be, for an $n$-tuple of scalars
$\lambda_{1},\dots,\lambda_{d}$,
\begin{equation}
\min_{\left\Vert \mathbf{v}\right\Vert =1}\max_{j}
\left\Vert X_{j}\mathbf{v}-\lambda_{j}\mathbf{v}
\right\Vert .
\label{eq:minimizeEigenErrors}
\end{equation}
If this quantity is small enough, often enough, we could declare these
observables to be approximately compatible.

This seems like a nearly impossible minimization, so we seek a proxy.
Commuting operators have a nice joint spectrum called the Clifford
spectrum, so we apply the definition of Clifford spectrum to tuples
of matrices with relatively small commutators and see what happens.
The resulting joint spectrum has nice theoretical properties, a beautiful
relation with $K$-theory, but remains difficult to compute numerically.
Generalizing pseudospectrum to what we call the Clifford pseudospectrum,
we find an efficiently computable approximation to 
Equation~\ref{eq:minimizeEigenErrors}. 
\begin{defn}
Suppose $X_{1}$ through $X_{d}$ are Hermitian matrices.  Let
$\Gamma_{1},\dots,\Gamma_{d}$ be any Hermitian representation of the
relations for $\mathcal{C\ell}_{d,0}(\mathbb{C})$,
meaning $\Gamma_{j}^{*}=\Gamma_{j}$, $\Gamma_{j}^{2}=1$ and 
$\Gamma_{j}\Gamma_{k}=-\Gamma_{k}\Gamma_{j}$
for $j\neq k$. The \emph{Clifford $\epsilon$-pseudospectrum} of
$(X_{1},\dots,X_{d})$ is
\[
\Lambda_{\epsilon}(X_{1},\dots,X_{d})
=
\left\{ 
\boldsymbol{\lambda}\in\mathbb{R}^{d}
\left|\,
\left\Vert \left(\sum\left(X_{j}-\lambda_{j}\right)\otimes\Gamma_{j}\right)^{-1}\right\Vert \geq\epsilon^{-1}
\right.
\right\} 
\]
with the convention $0^{-1}=\infty$ and $\left\Vert S^{-1}\right\Vert =\infty$
whenever S is singular. The \emph{Clifford spectrum} of $(X_{1},\dots,X_{d})$
is $\Lambda_{0}(X_{1},\dots,X_{d})$, also denoted $\Lambda(X_{1},\dots,X_{d})$.
The complement of the Clifford spectrum we call the 
\emph{Clifford resolvent set}.
\end{defn}

We will use the notation
\[
B(X_{1},\dots,X_{d})=\sum X_{j}\otimes\Gamma_{j}
\]
and
\begin{align*}
B_{\boldsymbol{\lambda}}(X_{1},\dots,X_{d}) & =B(X_{1}-\lambda_{1}I,\dots,X_{n}-\lambda_{d}I)\\
 & =B(X_{1},\dots,X_{n})-B(\lambda I,\dots,\lambda_{d}I)).
\end{align*}
For example, 
\[
\Lambda(X_{1},\dots,X_{d})
=
\left\{ 
\boldsymbol{\lambda}\in\mathbb{R}^{d}
\left|\, 
B_{\boldsymbol{\lambda}}(X_{1},\dots,X_{d})\mbox{ is singular}
\right.
\right\} .
\]

The representations of $\mathcal{C\ell}_{d,0}(\mathbb{C})$ are not
complicated, so it is routine to show this definition does not depend
on the choice of the $\Gamma_{j}$. When we get to $K$-theory we
will need to keep the matrix size as small as possible to avoid multiplicity
in the spectrum of $B_{\boldsymbol{\lambda}}(X_{1},\dots,X_{n})$.

\begin{lem}
\label{lem:create_approx_e_vectors} 
Suppose $X_{1},\dots,X_{d}$ are Hermitian $n$-by-$n$ matrices. If 
$\boldsymbol{\lambda}$ is in $\Lambda_{\epsilon}(X_{1},\dots,X_{d})$ then
there is a unit vector $\mathbf{v}$ in $\mathbb{C}^{n}$ with 
\[
\left\Vert X_{j}\mathbf{v}-\lambda_{j}\mathbf{v}\right\Vert 
\leq
\sqrt{\left\lceil \frac{d+1}{2}\right\rceil }
\sqrt{\epsilon^{2}+\sum_{j\neq k}\left\Vert \left[X_{j},X_{k}\right]\right\Vert }
\]
for all $j$. 
\end{lem}

\begin{proof}
Assume we have selected the $\Gamma_{j}$ in $\mathbf{M}_{g}(\mathbb{C})$
where $g$ is minimal, so $g=\left\lceil \frac{d+1}{2}\right\rceil $.
Without loss of generality, we assume $\boldsymbol{\lambda}$ equals
$\mathbf{0}$, and 
\[
\epsilon=\left\Vert \left(B(X_{1},\dots,X_{d}\right)^{-1}\right\Vert ^{-1}.
\]
Since $B(X_{1},\dots,X_{d})$ is Hermitian, $\epsilon$ has the alternate
description as the absolute value of the smallest eigenvalue of $B(X_{1},\dots,X_{d})$.
Let $\mathbf{z}$ be a corresponding unit eigenvector. Since
\begin{equation}
B(X_{1},\dots,X_{d})^{2}
=
\sum_{j}X_{j}^{2}\otimes I_{g}+\sum_{j\neq k}[X_{j},X_{k}]\otimes\Gamma_{j}\Gamma_{k}
\label{eq:eval_B_squared}
\end{equation}
we make the estimate 
\begin{align*}
\left\Vert \left(\sum X_{j}^{2}\otimes I_{g}\right)\mathbf{z}\right\Vert  
& \leq\left\Vert B(X_{1},\dots,X_{d})^{2}\mathbf{z}\right\Vert 
+
\left\Vert \left(\sum_{j\neq k}[X_{j},X_{k}]\otimes\Gamma_{j}\Gamma_{k}\right)\mathbf{z}\right\Vert \\
& \leq\epsilon^{2}+\sum_{j\neq k}\left\Vert \left[X_{j},X_{k}\right]\right\Vert .
\end{align*}
Let 
\[
\mathbf{z}=\left[\begin{array}{c}
\mathbf{z}_{1}\\
\vdots\\
\mathbf{z}_{g}
\end{array}\right]
\]
and let $r$ be an index maximizing $\left\Vert \mathbf{z}_{r}\right\Vert $,
so we have $\left\Vert \mathbf{z}_{r}\right\Vert \geq1/g$. 
Let $\mathbf{v}=\mathbf{z}_{r}/\left\Vert \mathbf{z_{r}}\right\Vert $.
Since
\[
\left\Vert \sum X_{j}^{2}\mathbf{z}_{r}\right\Vert 
\leq
\left\Vert \left(\sum X_{j}^{2}\otimes I_{g}\right)\mathbf{z}\right\Vert 
\]
we have
\[
\left\Vert \sum X_{j}^{2}\mathbf{v}\right\Vert 
\leq 
g\left\Vert \left(\sum X_{j}^{2}\otimes I_{g}\right)\mathbf{z}\right\Vert .
\]
Since $X_{\ell}^{2}\leq\sum_{j}X_{j}^{2}$ we find
\[
\left\langle X_{\ell}^{2}\mathbf{v},\mathbf{v}\right\rangle 
\leq
\left\langle \sum_{j}X_{j}^{2}\mathbf{v},\mathbf{v}\right\rangle 
\leq
\left\Vert \sum_{j}X_{j}^{2}\mathbf{v}\right\Vert 
\]
and so
\[
\left\langle X_{\ell}^{2}\mathbf{v},\mathbf{v}\right\rangle 
\leq 
g\left(\epsilon^{2}+\sum_{j\neq k}\left\Vert \left[X_{j},X_{k}\right]\right\Vert \right).
\]
\end{proof}

\begin{lem}
\label{lem:approx_evalReverse} Suppose $X_{1},\dots,X_{d}$ are Hermitian
$n$-by-$n$ matrices. If 
\[
\left\Vert X_{j}\mathbf{v}-\lambda_{j}\mathbf{v}\right\Vert \leq \epsilon
\]
for all $j$, then $\boldsymbol{\lambda}$ is in $\Lambda_{\epsilon'}(X_{1},\dots,X_{d})$
where 
\[
\epsilon' 
=
\sqrt{
\left
\lceil \frac{d+1}{2}\right\rceil ^{\frac{1}{2}} d\epsilon^{2}
+ \sum_{j\neq k}\left\Vert \left[X_{j},X_{k}\right]\right
\Vert }.
\]
\end{lem}

\begin{proof}
Again we use linearity to reduce to the case $\boldsymbol{\lambda}=\mathbf{0}$.
If $\left\Vert X_{j}\mathbf{v}\right\Vert \leq\epsilon$ for all $j$
then let 
\[
\mathbf{z}=\left[\begin{array}{c}
\mathbf{v}\\
\vdots\\
\mathbf{v}
\end{array}\right].
\]
Using Equation~\ref{eq:eval_B_squared} we find
\begin{align*}
\left\Vert \left(B(X_{1},\dots,X_{d})\right)^{2}\mathbf{z}\right\Vert  
& \leq\left\Vert \left(\sum_{j}X_{j}^{2}\otimes I_{g}\right)\mathbf{z}\right\Vert +\sum_{j\neq k}\left\Vert \left[X_{j},X_{k}\right]\right\Vert \\
& =\sqrt{g\left\Vert \sum X_{j}^{2}\mathbf{v}\right\Vert ^{2}}+\sum_{j\neq k}\left\Vert \left[X_{j},X_{k}\right]\right\Vert \\
& \leq\sqrt{g}\sum\left\Vert X_{j}^{2}\mathbf{v}\right\Vert +\sum_{j\neq k}\left\Vert \left[X_{j},X_{k}\right]\right\Vert \\
& \leq\sqrt{g}d\epsilon^{2}+\sum_{j\neq k}\left\Vert \left[X_{j},X_{k}\right]\right\Vert .
\end{align*}
This gives a lower bound on the norm of $\left(B(X_{1},\dots,X_{d})\right)^{2}$,
specifically
\[
\left\Vert \left(B(X_{1},\dots,X_{d})\right)^{-2}\right\Vert 
\geq
\left(\sqrt{g}d\epsilon^{2}+\sum_{j\neq k}\left\Vert \left[X_{j},X_{k}\right]\right\Vert \right)^{-1}.
\]
Since $B(X_{1},\dots,X_{d})$ is Hermitian, we conclude
\[
\left\Vert \left(B(X_{1},\dots,X_{d})\right)^{-1}\right\Vert ^{-1}
\leq
\sqrt{\sqrt{g}d\epsilon^{2}+\sum_{j\neq k}\left\Vert \left[X_{j},X_{k}\right]\right\Vert }.
\]
\end{proof}

\begin{rem}
Lemmas \ref{lem:create_approx_e_vectors} and \ref{lem:approx_evalReverse}
tell us that for almost commuting Hermitian matrices $X_{1},\dots,X_{d}$
we can get an approximation to the quantity in Equation~\ref{eq:minimizeEigenErrors}
by computing 
\[
\left\Vert \left(B_{\boldsymbol{\lambda}}(X_{1},\dots,X_{d})\right)^{-1}\right\Vert ^{-1}.
\]
In a numerical setting, we can compute this easily. For example, we
can compute the absolute value of the eigenvalue of $B(X_{1},\dots,X_{d})$
that is closest to zero. This matrix is Hermitian, and typically sparse,
so standard algorithms work well for modest matrix sizes. The algorithms
typically compute an associated (approximate) eigenvalue, so we have
a way to construct vectors that come close to the minimum in 
Equation~\ref{eq:minimizeEigenErrors}.  As we push the matrix sizes
larger, we will need to do better. Still,
estimating the norm of an inverse is a fairly standard problem in
numerical analysis. One issue is that it is hard to differentiate
an eigenvalue at zero from one close to zero. This is why we turn
to the pseudospectrum. If we are computing the function
\[
\boldsymbol{\lambda}
\mapsto
\left\Vert \left(B_{\boldsymbol{\lambda}}(X_{1},\dots,X_{d})\right)^{-1}\right\Vert ^{-1}
\]
we need to set a value $\epsilon$ just above zero and regard all
values below that as equal. This is very reasonable, as we are modeling
simultaneous approximate measurement when true simultaneous measurement
is impossible.
\end{rem}

\begin{example}
If $X_{1},\dots,X_{d}$ are commuting Hermitian matrices then $\Lambda(X_{1},\dots,X_{d})$
equals the usual joint spectrum. This is an immediate corollary of
Lemma~\ref{lem:create_approx_e_vectors} and Lemma~\ref{lem:approx_evalReverse}. 
\end{example}

\begin{example}
If $A$ and $B$ are Hermitian, then $\left(\lambda_{1},\lambda_{2}\right)$
is in $\Lambda(A,B)$ if and only if $\lambda_{1}+i\lambda_{2}$ is
in the spectrum of $A+iB$. So the Clifford spectrum of a pair of
Hermitian matrices is finite. For positive $\epsilon$ we can show
that$\left(\lambda_{1},\lambda_{2}\right)$ is in $\Lambda_{\epsilon}(A,B)$
if and only if $\lambda_{1}+i\lambda_{2}$ is in the usual pseudospectrum
of $A+iB$. However need the convention 
\[
\sigma_{\epsilon}(Y)
=
\left\{ \alpha\in\mathbb{C}\left|\,\left\Vert \left(\alpha-Y\right)^{-1}\right\Vert \geq\epsilon^{-1}\right.\right\} 
\]
and not the convention with strict inequality, as in the excellent
book \cite{TrefethenEmbree} by Trefethen and Embree. To see the connection,
we temporarily use
\[
\Gamma_{1}=\left[\begin{array}{cc}
0 & 1\\
1 & 0
\end{array}\right],\quad\Gamma_{2}=\left[\begin{array}{cc}
0 & i\\
-i & 0
\end{array}\right]
\]
so that 
\[
-B_{(\lambda_{1},\lambda_{2})}(A,B)=\left[\begin{array}{cc}
0 & \left(\lambda_{1}+i\lambda_{2}\right)-\left(A+iB\right)\\
\left(\left(\lambda_{1}+i\lambda_{2}\right)-\left(A+iB\right)\right)^{*} & 0
\end{array}\right].
\]

Often the better choices here are 
\[
\Gamma_{1}=\left[\begin{array}{cc}
0 & 1\\
1 & 0
\end{array}\right],\quad\Gamma_{2}=\left[\begin{array}{cc}
1 & 0\\
0 & -1
\end{array}\right]
\]
as this keeps 
\[
B_{(\lambda_{1},\lambda_{2})}(A,B)=\left[\begin{array}{cc}
B-\lambda_{2} & A-\lambda_{1}\\
A-\lambda_{1} & -B+\lambda_{2}
\end{array}\right]
\]
real when $A$ and $B$ are real. Then we are able to produce real
joint approximate eigenvalues for $A$ and $B$ by finding near null
vectors of $B_{(\lambda_{1},\lambda_{2})}(A,B)$.
\end{example}

Next an example where the Clifford spectrum is an infinite set. For
$d=3$ the clear choice for the $\Gamma_{j}$ is $\Gamma_{1}=\sigma_{x}$,
$\Gamma_{2}=\sigma_{y}$, $\Gamma_{3}=\sigma_{z}$ so that 
\[
B(X,Y,Z)=\left[\begin{array}{cc}
Z & X-iY\\
X+iY & -Z
\end{array}\right]
\]
as was done in previous work with Hastings \cite{HastingsLoringWannier}. 

\begin{example}
\label{exa:PauliSpin} 
A nice example, computed by Kisil \cite{KisilCliffordSpectrum},
shows us that the Clifford spectrum for three Hermitian matrices is
radically different from the Clifford spectrum of two Hermitian matrices
(as defined below), as it need not be a finite set. We compute 
$\Lambda(\sigma_{x},\sigma_{y},\sigma_{z})$ with help from a symbolic
algebra package. The ``characteristic polynomial'' here is 
\begin{align*}
 & \det\left(B(\sigma_{x}-rI,\sigma_{y}-sI,\sigma_{z}-tI)\right)\\
 & \quad=\det\left(\left[\begin{array}{cccc}
1-t & 0 & -r+is & 0\\
0 & -1-t & 2 & -r+is\\
-r-is & 2 & -1+t & 0\\
0 & -r-is & 0 & 1+t
\end{array}\right]\right)\\
 & \quad=(r^{2}+s^{2}+t^{2}+1)^{2}-4.
\end{align*}
This means $\Lambda(\sigma_{x},\sigma_{y},\sigma_{z})$ is the unit
sphere. 
\end{example}

While investigating D-branes, Berenstein Malinowski \cite{berenstein2012matrix}
took the preceding example further. In that setting, the position
observables do not commute. They looked at higher spin representations
and computed the Clifford spectrum, again a sphere. In fact they were
interested in a subset of the Clifford spectrum that needs some form
of $K$-theory for its definition. 

Where an index, and eventually $K$-theory, arise is easily seen in
Example~\ref{exa:PauliSpin}. Let us examine what is going on at
two points in the Clifford resolvent set, the origin and $(0,0,2)$.
We find that
\[
B_{\mathbf{0}}(\sigma_{x},\sigma_{y},\sigma_{z})=\left[\begin{array}{cccc}
1 & 0 & 0 & 0\\
0 & -1 & 2 & 0\\
0 & 2 & -1 & 0\\
0 & 0 & 0 & 1
\end{array}\right]
\]
which has a single eigenvalue at $-3$ and a triple eigenvalue at
$1$. On the other hand 
\[
B_{(0,0,2)}(\sigma_{x},\sigma_{y},\sigma_{z})=\left[\begin{array}{cccc}
-3 & 0 & 0 & 0\\
0 & -3 & 2 & 0\\
0 & 2 & 1 & 0\\
0 & 0 & 0 & 3
\end{array}\right]
\]
has spectrum
\[
\left\{ -1-\sqrt{8},-1,-1+\sqrt{8},3\right\} 
\]
and so the same number of positive and negative eigenvalues. As we
vary $\boldsymbol{\lambda}$ the eigenvalues move continuously. It
follows that for any $\boldsymbol{\lambda}$ inside the unit sphere
$B_{\boldsymbol{\lambda}}(\sigma_{x},\sigma_{y},\sigma_{z})$ will
have just one positive eigenvalue. For $\boldsymbol{\lambda}$ outside
the unit sphere $B_{\boldsymbol{\lambda}}(\sigma_{x},\sigma_{y},\sigma_{z})$
will have exactly two positive eigenvalues. The Clifford resolvent
set contains information and we will see that from a computation of
$B_{\boldsymbol{\lambda}}(X_{1},\dots,X_{d})$ at a single point we
can make predictions on the size of the Clifford spectrum. This is
the mathematical essence of bulk-edge correspondence.

We now require that our $\Gamma_{j}$ are selected in matrices of
minimal size. In fact, let us consider $d=3$ and use the Pauli spin
matrices, as above. Recall that for an invertible Hermitian matrix
$Q$ the \emph{signature} of  $Q$ is the number of positive eigenvalues
(with multiplicity) minus the number of negative eigenvalues of $Q$,
denoted $\mathrm{Sig}(Q)$. Note the signature is always even for
even size matrices.

\begin{defn}
If $\boldsymbol{\lambda}$ is not in $\Lambda(X,Y,Z)$ then the index
of this triple at $\boldsymbol{\lambda}$ is
\[
\mathrm{Ind}_{\boldsymbol{\lambda}}(X,Y,Z)=\frac{1}{2}\mathrm{Sig}\left(B_{\boldsymbol{\lambda}}(X,Y,Z)\right)
\]
which is in the abelian group $\mathbb{Z}$.
\end{defn}

\begin{rem}
We us primarily the notation of pure mathematics.  In particular $*$ 
indicates the conjugate transpose of a matrix, and is a special instance
 of the $*$ operation in a $C^*$-algebra.
\end{rem}

\begin{example}
\label{exa:CherInsulator} Consider a finite model of a two-dimensional
Chern insulator on square lattice. That is, with zero for boundary
conditions. The Hamiltonian we consider a tight binding model, where
there are two orbital types at each site on a square lattice. We have
creation operators $c_{m,n,P}$ and $c_{m,n,S}$ at site $(m,n)$
in band P or S. Let $c_{m,n}$ be the sum of these two types of create
at the same site. The periodic Hamiltonian is 
\begin{align*}
H_{\mathrm{per}} 
& =\sum_{m,n}c_{m,n}^{*}\left(-\sigma_{z}-\sigma_{y}\right)c_{m+1,n}+h.c.\\
& \quad+\sum_{m,n}c_{m,n}^{*}\left(-\sigma_{z}-i\sigma_{y}\right)c_{m,n+1}+h.c.\\
& \quad+\sum_{m,n}c_{m,n}^{*}\left(\left(3+\mu_{m,n}\right)\sigma_{z}\right)c_{m,n}
\end{align*}
where $\mu_{m,n}$ is drawn with uniform distribution from 
$\left[-\tfrac{N}{2},\tfrac{N}{2}\right]$ where $N$ sets the disorder
level. This is the model used for a Chern insulator as part of the numerical
study done with Hastings \cite{LorHastHgTe},
which was essentially the spin-up only part of the model for an HgTe
quantum wells given in \cite{konig2007quantum}. If we use lattice
position (roughly \r{A}ngstr\"{o}ms) in defining our position operators,
we find $\left\Vert \left[H,X\right]\right\Vert $
and $\left\Vert \left[H,Y\right]\right\Vert $
rather large, about $6$. We work with the triple $(\eta X,\eta Y,H)$,
although we plot our results using lattice units. For this example,
$\eta=0.5$ was selected as a value for which the computed approximate
eigenvectors we spread out roughly one nanometer in position. 
We calculated the $\epsilon$-pseudospectrum, and also the index at many
positions at the Fermi level. For better viewing, the $X$ coordinate
of the $\boldsymbol{\lambda}$ was truncated to $[-2,2]$ and energy
coordinate to $[-2.5,2.5]$. The full energy spectrum is roughly $[-7,7]$.
This portion of the pseudospectrum and labeled resolvent is shown in
Figures \ref{fig:Chern-insulator-near-Fermi-0}-\ref{fig:Chern-insulator-near-Fermi-10}
with an increasingly large random disorder. The pseudospectrum is
calculated at 5 grid points per unit and $\epsilon=0.05$. 
\end{example}

\begin{figure}
\includegraphics[clip,scale=1.0,bb=160 290 460 480]{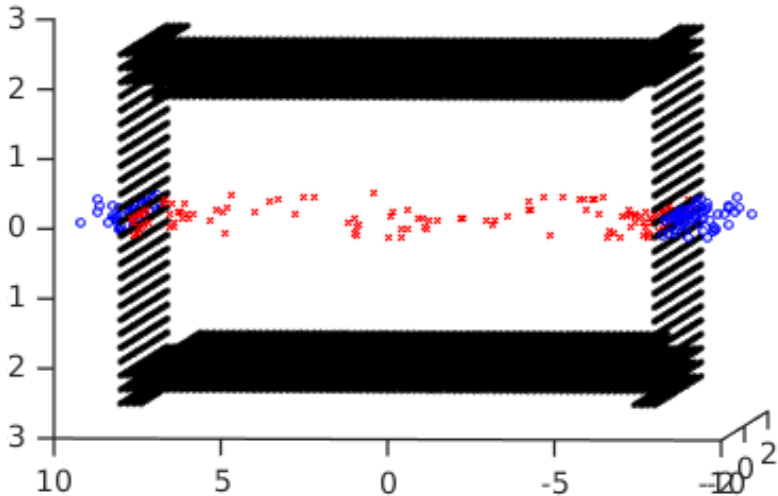}
\caption{A Chern insulator on a $18$-by-18 lattice with no disorder. 
\label{fig:Chern-insulator-near-Fermi-0}}
\end{figure}

\begin{figure}
\includegraphics[clip,scale=1.0,bb=160 290 460 480]{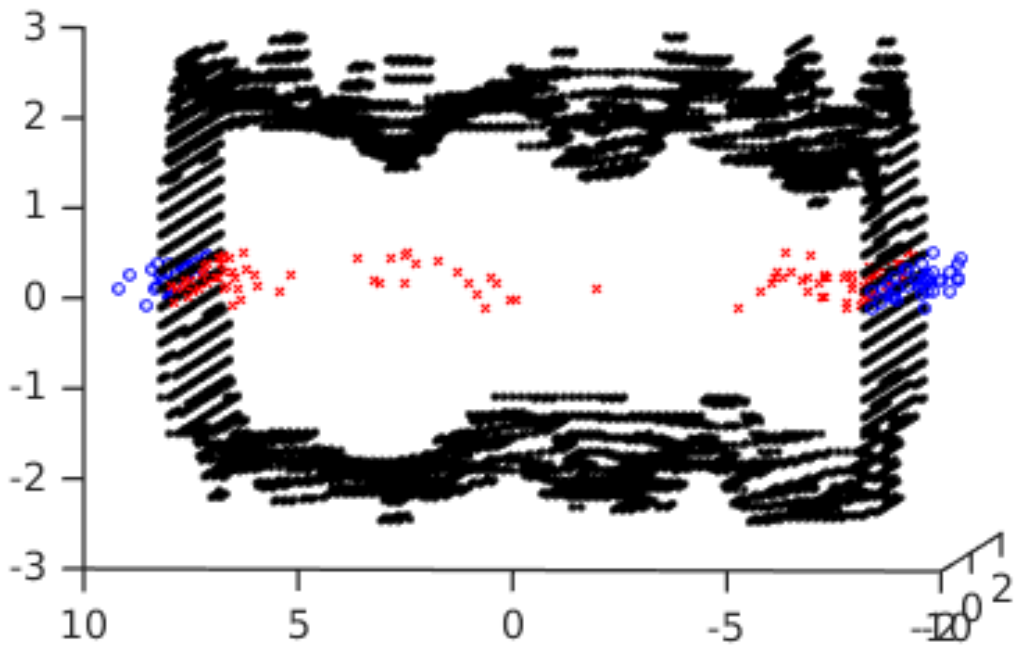}
\caption{
A Chern insulator on a $18$-by-18 lattice with disorder at 4. \label{fig:Chern-insulator-near-Fermi-4}}
\end{figure}

\begin{figure}
\includegraphics[clip,scale=1.0,bb=160 290 460 480]{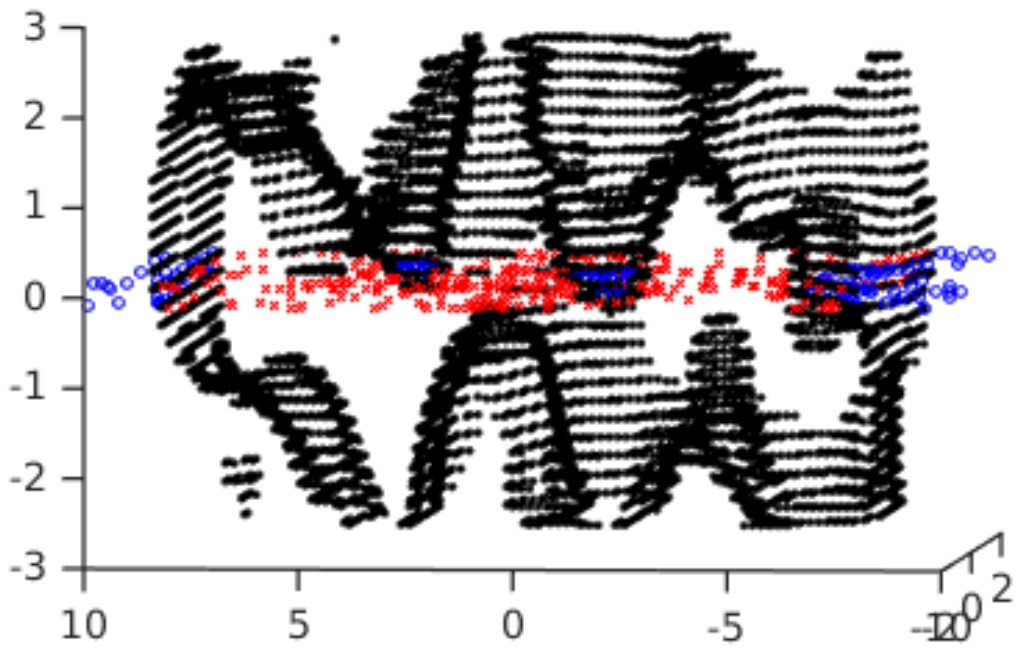}
\caption{A Chern insulator on a $18$-by-18 lattice with disorder at 10. 
\label{fig:Chern-insulator-near-Fermi-10}}
\end{figure}

\section{Almost commuting matrices, the ten-fold way}

We have many potential sources of almost commuting matrices, but now
focus on situations most relevant to topological insulators. A lattice
model of a $D$-dimensional topological or ordinary insulator needs
D+1 Hermitian matrices to be described. The $D$ position operators
$X_{j}$ will commute with each other and almost commute with the
Hamiltonian $H$. In most situations our formulas will work just as
well if the $X_{j}$ almost commute, so we do not always require the
$X_{j}$ to exactly commute.

We now consider the ten symmetry classes in the Atland-Zirnbauer 
\cite{altland1997nonstandard} classification.
Depending on the symmetry class, we may have up to three symmetries.
Time reversal will be denoted $\mathcal{T}$ and be antiunitary and
commute with all $D+1$ matrices. Particle-hole conjugation will also
be anti-unitary, will commute with position and anticommute with the
Hamiltonian. The symmetry $S$ will be unitary, commute with position
and anticommute with $H$. If all three are present, then $S$ is
the product of the other two, which commute. All symmetries are of
order two.

\begin{table}
\begin{tabular}{|>{\raggedright}p{0.7in}|l|l|l|}
\hline 
\multirow{2}{0.7in}{Cartan class} & \multirow{2}{*}{$\mathcal{T}$} & \multirow{2}{*}{$\mathcal{C}$} & \multirow{2}{*}{$S$}\\
 &  &  & \\
\hline 
\multicolumn{1}{|>{\raggedright}p{0.7in}}{\emph{Complex}} & \multicolumn{1}{l}{} & \multicolumn{1}{l}{} & \\
\hline 
A & --- & --- & ---\\
\hline 
AIII & --- & --- & $\checkmark$\\
\hline 
\multicolumn{1}{|>{\raggedright}p{0.7in}}{\emph{Real}} & \multicolumn{1}{l}{} & \multicolumn{1}{l}{} & \\
\hline 
AI & $\mathcal{T}\circ\mathcal{T}=I$ & --- & ---\\
\hline 
BDI & $\mathcal{T}\circ\mathcal{T}=I$ & $\mathcal{C}\circ\mathcal{C}=I$ & $\checkmark$\\
\hline 
D & --- & $\mathcal{C}\circ\mathcal{C}=I$ & ---\\
\hline 
DIII & $\mathcal{T}\circ\mathcal{T}=-I$ & $\mathcal{C}\circ\mathcal{C}=I$ & $\checkmark$\\
\hline 
AII & $\mathcal{T}\circ\mathcal{T}=-I$ & --- & ---\\
\hline 
CII & $\mathcal{T}\circ\mathcal{T}=-I$ & $\mathcal{C}\circ\mathcal{C}=-I$ & $\checkmark$\\
\hline 
C & --- & $\mathcal{C}\circ\mathcal{C}=-I$ & ---\\
\hline 
CI & $\mathcal{T}\circ\mathcal{T}=I$ & $\mathcal{C}\circ\mathcal{C}=-I$ & $\checkmark$\\
\hline 
\end{tabular}

\vspace{0.2in}

\caption{The symmetry classes. The mark $\checkmark$ means this symmetry exists,
and $SHS=-H$ and $SX_{j}S=X_{j}$ for a unitary $S$, with $H$ the
Hamiltonian and $X_{j}$ any of the position observables. The mark
--- means this symmetry does not exists. If the other two symmetries
exist then $S=\mathcal{T}\circ\mathcal{C}=\mathcal{C}\circ\mathcal{T}$.
All other formulas indicate the symmetry exists, with the stated conditions
holding, with $SHS=-H$ and $SX_{j}S=X_{j}$, or $\mathcal{T}\circ H=H\circ\mathcal{T}$
and $\mathcal{T}\circ X_{j}=X_{j}\circ\mathcal{T}$, or $\mathcal{C}\circ H=-H\circ\mathcal{C}$
and $\mathcal{C}\circ X_{j}=X_{j}\circ\mathcal{C}$.
\label{tab:The-symmetry-classes}
}
\end{table}

Following Kitaev \cite{kitaevEdgeNotes}, we will focus on situations
where there is a localized spectral gap. A minimal interpretation of this is 
\[
\mathbf{0}\notin\Lambda(X_{1},\dots,X_{D},H).
\]
However, we are dealing with approximate measurement, so a better
definition is generally
\[
\left\Vert B\left(X_{1},\dots,X_{D},H\right)^{-1}\right\Vert \leq\delta_{1}
\]
and $\left\Vert \left[X_{j},X_{k}\right]\right\Vert \leq\delta_{2}$
and $\left\Vert \left[X_{j},H\right]\right\Vert \leq\delta_{2}$ .
For each choice of $\delta_{1}$ and $\delta_{2}$ we are specifying
a potentially useful collection of systems.

In lower dimensions we can hope to classify such systems up to homotopy
and identify computable invariants to detect these homotopy classes.
As we move to higher dimensions we need to allow for stabilization
by adding on trivial systems.

\begin{table}
\begin{tabular}{|>{\centering}p{0.7in}|c|c|c|c|c|c|c|c|}
\hline 
\multicolumn{1}{|>{\centering}p{0.7in}}{\multirow{2}{0.7in}{Cartan class}} & \multicolumn{8}{c|}{Dimension}\\
\cline{2-9} 
 & \multicolumn{1}{c}{0} & \multicolumn{1}{c}{1} & \multicolumn{1}{c}{2} & \multicolumn{1}{c}{3} & \multicolumn{1}{c}{4} & \multicolumn{1}{c}{5} & \multicolumn{1}{c}{6} & 7\\
\hline 
\multicolumn{1}{|>{\centering}p{0.7in}}{\emph{Complex}} & \multicolumn{1}{c}{} & \multicolumn{1}{c}{} & \multicolumn{1}{c}{} & \multicolumn{1}{c}{} & \multicolumn{1}{c}{} & \multicolumn{1}{c}{} & \multicolumn{1}{c}{} & \\
\hline 
A & $\mathbb{Z}$ & $0$ & $\mathbb{Z}$ & $0$ & $\mathbb{Z}$ & $0$ & $\mathbb{Z}$ & $0$\\
\hline 
AIII & $0$ & $\mathbb{Z}$ & $0$ & $\mathbb{Z}$ & $0$ & $\mathbb{Z}$ & $0$ & $\mathbb{Z}$\\
\hline 
\multicolumn{1}{|>{\centering}p{0.7in}}{\emph{Real}} & \multicolumn{1}{c}{} & \multicolumn{1}{c}{} & \multicolumn{1}{c}{} & \multicolumn{1}{c}{} & \multicolumn{1}{c}{} & \multicolumn{1}{c}{} & \multicolumn{1}{c}{} & \\
\hline 
AI & $\mathbb{Z}$ & $0$ & $0$ & $0$ & $\mathbb{Z}$ & $0$ & $\mathbb{Z}_{2}$ & $\mathbb{Z}_{2}$\\
\hline 
BDI & $\mathbb{Z}_{2}$ & $\mathbb{Z}$ & $0$ & $0$ & $0$ & $\mathbb{Z}$ & $0$ & $\mathbb{Z}_{2}$\\
\hline 
D & $\mathbb{Z}_{2}$ & $\mathbb{Z}_{2}$ & $\mathbb{Z}$ & $0$ & $0$ & $0$ & $\mathbb{Z}$ & $0$\\
\hline 
DIII & $0$ & $\mathbb{Z}_{2}$ & $\mathbb{Z}_{2}$ & $\mathbb{Z}$ & $0$ & $0$ & $0$ & $\mathbb{Z}$\\
\hline 
AII & $\mathbb{Z}$ & $0$ & $\mathbb{Z}_{2}$ & $\mathbb{Z}_{2}$ & $\mathbb{Z}$ & $0$ & $0$ & $0$\\
\hline 
CII & $0$ & $\mathbb{Z}$ & $0$ & $\mathbb{Z}_{2}$ & $\mathbb{Z}_{2}$ & $\mathbb{Z}$ & $0$ & $0$\\
\hline 
C & $0$ & $0$ & $\mathbb{Z}$ & $0$ & $\mathbb{Z}_{2}$ & $\mathbb{Z}_{2}$ & $\mathbb{Z}$ & $0$\\
\hline 
CI & $0$ & $0$ & $0$ & $\mathbb{Z}$ & $0$ & $\mathbb{Z}_{2}$ & $\mathbb{Z}_{2}$ & $\mathbb{Z}$\\
\hline 
\end{tabular}

\vspace{0.2in}

\caption{The range of the invariants \cite{kitaev-2009,ryu2010topological}.
\label{tab:range-of-invariants}
}
\end{table}

In addition to homotopy questions, we can ask of a given system is
close to another system that is in the atomic limit, where the Hamiltonian
commutes with all the position operators. This is then a special case
of a mathematical question. Given $D+1$ almost commuting Hermitian
matrices $X_{1},\dots,X_{D+1}$ in specific AZ class that almost commute,
are there nearby commuting Hermitian matrices in the same class that
commute. To be a serious question this must be posed in a way that
is uniform for all matrix sizes.  That is, in Theorem~\ref{thm:Lin}
the selected $\delta$ must work for all choices for the matrix
size $n$.

One instance of this, for dimension 2 in class D, asks the following.
Given two real symmetric and one imaginary antisymmetric matrices that
almost commute, can these can be uniformly approximated by
commuting matrices, again with two being real symmetric and being
imaginary antisymmetric. The answer is
no, with an obstruction in $\mathbb{Z}$, as indicated
in table~\ref{tab:range-of-invariants}. 

We tend to prefer describing the symmetries in terms of operations
on matrices \cite{LoringQuaternions}. So we work with the dual operation
$\sharp$ that is derived from fermionic time-reversal by
\[
Q^{\sharp}=\mathcal{T}^{-1}\circ Q^{*}\circ\mathcal{T}
\]
were 
\[
\mathcal{T}\left[\begin{array}{c}
\mathbf{v}\\
\mathbf{w}
\end{array}\right]=\left[\begin{array}{c}
-\overline{\mathbf{w}}\\
\overline{\mathbf{v}}
\end{array}\right].
\]
In block form,
\[
\left[\begin{array}{cc}
A & B\\
C & D
\end{array}\right]^{\sharp}=\left[\begin{array}{cc}
D^{\mathrm{T}} & -B^{\mathrm{T}}\\
-C^{\mathrm{T}} & A^{\mathrm{T}}
\end{array}\right].
\]

The first two columns in Table~\ref{tab:range-of-invariants} are
unique. These invariants are just a reflection of the homotopy classes
of such locally-gapped systems. In these columns the invariants do
not represent obstructions to a system being close to another system
in the atomic limit. The mathematics behind this statement is nontrivial.
It says that approximately measuring two incompatible observables
simultaneously is very different from doing so with three or more. 

\begin{thm} \label{thm:Lin}
\textup{(Lin's Theorem, 1995)} For any $\epsilon>0,$ there exists
$\delta>0$ such that whenever $n\in\mathbb{N}$ and two self-adjoint
matrices $H$ and $X$ in $\mathbf{M}_{n}(\mathbb{C})$ satisfy
$\left\Vert H\right\Vert \leq1$ and $\left\Vert X\right\Vert \leq1$ and
\[
\left\Vert HX-XH\right\Vert \leq\delta,
\]
there exists a pair of self-adjoint matrices $H_{1}$ and $X_{1}$
in $\mathbf{M}_{n}(\mathbb{C})$ such that 
\[
\left\Vert H_{1}-H\right\Vert \leq\epsilon,\quad\left\Vert X_{1}-X\right\Vert \leq\epsilon
\]
and $H_{1}X_{1}=X_{1}H_{1}.$
\end{thm}

Lin's original proof \cite{LinAlmostCommutingHermitian} is difficult,
so perhaps a better starting point in the literature is \cite{friis1996almost}.
However, Ogata \cite{ogata2013approximating} adapted Lin's original
proof to work with more than two almost commuting Hermitian matrices
in a special case involving macroscopic observables. These observables
are multiparticle averages that avoid the $K$-theory in columns two
and above in Table~\ref{tab:range-of-invariants}.

\begin{conjecture}
Lin's Theorem remains true of we assume that $(X,H)$ is in any Atland-Zirnbauer
symmetry class and we require that $(X_{1},H_{1})$ be in the same
symmetry class.
\end{conjecture}

Joint work with S\o rensen \cite{LorSorensenDisk,LorSorensenOrtho}
proved that this conjecture is valid in classes AI, D, AII and C.
Of course Lin dealt with class A, leaving the conjecture open on the
five classes involving two antiunitary symmetries. Even in the complex
case, research into Lin's theorem is not over, in particualar looking
at quantitative versions, as in \cite{hastings-2008,kachkov_Safa_DistanceToNormal}.

The only serious consequence of Lin's theorem we explore here is
the following remark about $1D$ systems.  However, there are
expected to be results regarding the classification of 2D systems,
along the lines of the results in \cite{LorSorensenOrtho}.

\begin{rem}
Consider two almost commuting Hermitian operators $H$ and $X$.   The Clifford
spectrum of this pair will be a perturmation of the Clifford
spectrum of a commuting pair, which is a finite set.  
The Clifford pseudo spectrum computed for 1D systems has
typically been a collection of small disconnected regions.
The computations done have been limited, but it is expected
that the pseudo spectrum of a 1D system look very different from
the mutated spheres we see for 2D systems.
\end{rem}

\section{Dimension zero, index formulas\label{sec:Dimension-zero}}

Many approaches to the $K$-theory of topological insulators rely
on spectrally flattened Hamiltonian, or equivalently the Fermi
projector.  For example, see
\cite{BellissardNCGquantumHall,essin2007topological,Mondr_Prodan_AIII_1D,prodan2011disordered}. 
One approach that avoids this is the scattering matrix
approach \cite{fulga2011scattering,sbierski3DtopInsScattering}.
There are many numerical issues
related to working with matrices with high degeneracy in the spectrum,
so we prefer to work directly with the full class of invertible Hermitian
matrices.

The homotopy classification in dimension zero is rather standard.
However the formulas for the invariants are not all standard.
Our invariants are only designed to work for finite models, but
it is anticipated that there will be connections with indexes
defined to work in infinite volume, such as 
\cite{BellisardNthChern,Mondr_Prodan_AIII_1D,prodan2011disordered}.

\begin{thm}
In each of the ten Atland-Zirnbauer symmetry classes, two invertible
$n$-by-$n$ Hermitian matrices are homotopic via invertibles in that
class if and only if the values of an index (in the group $0$ or
$\mathbb{Z}$ or $\mathbb{Z}_{2}$) for each matrix are equal. This
index can be computed in $O(n^{3})$ time. The index for each class
is listed below.
\end{thm}

An important consideration is how these invariants can be computed
more quickly than $O(n^{3})$ when the matrices are sparse. We will
give brief remarks on sparse algorithms below.

\subsection{Class A in 0D\label{sub:Class-A-in-zero}}

There are no symmetries, except $H=H^{*}$ in $\mathrm{GL}_{n}(\mathbb{C})$.
The index is
\[
\tfrac{1}{2}\mathrm{Sig}\left(H\right)
\]
and it is just a variation on the spectral theorem that this classifies
such matrices up to homotopy.
\begin{rem}
\label{rem:complexSignature} The signature can be computed using
the LDLT decomposition, which finds a unit lower triangular matrix
$L$ and a block diagonal matrix $D$ with $1$-by-1 and $2$-by-2
blocks so that $H=LDL^{*}.$ By Sylvester's law of inertia, 
$\mathrm{Sig}\left(H\right)=\mathrm{Sig}\left(D\right)$
and the signature of $D$ can be found in linear time. Since $LDLT$
is $O(n^{3})$ we are done in this case. If $H$ is sparse, there
is a readily available sparse version of the LDLT algorithm \cite{DuffSparseFactorSym}. 
\end{rem}

\subsection{Class AI in 0D}

We now have $H$ real. We can view that as the added symmetry $H^{\mathrm{T}}=H$.
The index is again 
\[
\tfrac{1}{2}\mathrm{Sig}\left(H\right)
\]
and the algorithms mentioned in \S \ref{sub:Class-A-in-zero} are
available also in the real case. The essential fact in proving that
this invariant classifies is that $H$ can be factored as $H=U^{*}DU$
with $D$ diagonal with decreasing diagonal terms and $U$ being real
orthogonal with determinant one.

\subsection{Class BDI in 0D}

The symmetries here can be taken to be $H^{\mathrm{T}}=H$ and $H\Gamma=-\Gamma H$
for $H=H^{*}$ in $\mathrm{GL}_{n}(\mathbb{C})$ with 
\[
\Gamma=\left[\begin{array}{cc}
I & 0\\
0 & -I
\end{array}\right].
\]
That is,
\[
H=\left[\begin{array}{cc}
0 & C\\
C^{\mathrm{T}} & 0
\end{array}\right]
\]
for $C$ an invertible real matrix. The index in this case is
\[
\mathrm{sign}\left(\det\left(C\right)\right)
\]
in $\mathbb{Z}_{2}=\{\pm1\}$. Recall that two real invertible matrices
can be path connected if and only if their determinants are of the
same sign.
\begin{rem}
There is a noncanonical choice to be made here, specifically a real
unitary from $\mathbb{C}^{n}\oplus0\mapsto0\oplus\mathbb{C}$. 
\end{rem}

\begin{rem}
The sign of the determinant can be computed using the LU decomposition,
which finds a lower triangular matrix $L$ and an upper triangular
matrix $R$ so that $C=LR.$ Then
\[
\mathrm{sign}\left(\det\left(C\right)\right)
=
\prod\mathrm{sign}(L_{jj})\prod\mathrm{sign}(R_{jj})
\]
which avoids the underflow and overflow associated with computing
determinants of large matrices. Since $LU$ is $O(n^{3})$ we are
done in this case. If $H$ is sparse, there is are readily available
sparse versions of the LU algorithm \cite{DavisSparseLU}. 
\end{rem}

\subsection{Class D in 0D}

The symmetry here can be taken to be $H^{\mathrm{T}}=-H$ for $H=H^{*}$
in $\mathrm{GL}_{2n}(\mathbb{C})$. The eigenvalues of the real, normal
matrix $iH$ will come in conjugate pairs so its determinant is positive.
This makes the Pfaffian real, and our invariant is 
\[
\mathrm{sign}\left(\mathrm{Pf}\left(iH\right)\right).
\]
The homotopy classification can be understood here in terms of the
factorization \cite[Theorem 9.4]{HastLorTheoryPractice} of $H$ as 
$H=UDU^{*}$ where $U$ is real orthogonal and and $D$ is diagonal.

\begin{rem}
The sign of the Pfaffian of $K=iH$ can be computed using a decomposition
$K=LDL^{\mathrm{T}}$ where $D$ is tridiagonal. If $H$ is banded,
which will happen when working with derived Hamiltonians based on
1D systems, one can use software for Pfaffians by 
Wimmer \cite{wimmer2011efficient}.  There is an algorithm, but no 
available software, for the more general case of sparse, real skew-symmetric 
matrices \cite{DuffSparseFactorSkew}.
\end{rem}

\subsection{Class AII, 0D}

We now have $H$ self-dual, so $H^{*}=H^{\sharp}=H$. The index is
\[
\frac{1}{4}\mathrm{Sig}\left(H\right)
\]
due to Kramer's doubling. The homotopy classification is best understood
here in terms of the factorization of \cite[Theorem 4.6]{HastLorTheoryPractice},
$H=UDU^{*}$ where $U$ is a symplectic unitary and $D$ is diagonal
with $D_{j,,j}=D_{j+n,,j+n}$.

\section{Dimension one, index formulas\label{sec:Dimension-one}}

Although Lin's theorem is deep, we can avoid it when we have two incompatible
observables if we are willing to study systems up to homotopy. Given
$H$ and $X$ Hermitian matrices, one form of our local gap condition
\[
\left[\begin{array}{cc}
0 & \left(X+iH\right)^{*}\\
X+iH & 0
\end{array}\right]\mbox{ is invertible}
\]
translates to the condition $W=X+iH$ is invertible. It is easy to
deform an invertible to a unitary by the path
\[
W_{t}=W\left(W^{*}W\right)^{t}
\]
for $t$ in $[0,1]$. Extracting the Hermitian and anti-Hermitian
parts $X_{t}=\tfrac{1}{2}W_{t}^{*}+\tfrac{1}{2}W_{t}$ and 
$H_{t}=\tfrac{i}{2}W_{t}^{*}-\tfrac{i}{2}W_{t}$
we get a path to a system the atomic limit. At all points on the path
the commutator $\left[X_{t},H_{t}\right]$ will remain small if the
initial commutator is assumed sufficiently small. One can check that
this construction respects the various symmetries. For example, in
class C we have $H^{\sharp}=-H$ and $X^{\sharp}=-X$. This implies
$W^{\sharp}=W^{*}$. Functional calculus commutes with $\sharp$ so
\[
\left(W_{t}\right)^{\sharp}=\left(W^{*}W\right)^{t}W^{*}=\left(W_{t}\right)^{*}.
\]
Finally
\[
X_{t}^{\sharp}=\tfrac{1}{2}W_{t}^{*}+\tfrac{1}{2}W_{t}=X_{t}
\]
and
\[
H_{t}^{\sharp}=\tfrac{1}{2}W_{t}+\tfrac{1}{2}W_{t}^{*}=-H_{t}.
\]

By this homotopy argument, we can simply check that a formula is invariant
under homotopy and that correctly classifies systems in the atomic
limit, at least up to homotopy.

\subsection{Class AIII in 1D \label{sub:Class-AIII-1D}}

We can assume we have $H=H^{*}$ and $X=X^{*}$ in $\mathbf{M}_{n}(\mathbb{C})$,
with $X+iH$ assumed to be invertible. With $ $ 
\[
\Gamma=\left[\begin{array}{cc}
I & 0\\
0 & -I
\end{array}\right]
\]
defining a grading, we have $H\Gamma=-\Gamma H$ and $X\Gamma=\Gamma X$.
This means that $\left(X+iH\right)\Gamma$ is Hermitian. It is invertible
because $\Gamma$ is unitary. The index we use here is
\[
\tfrac{1}{2}\mathrm{Sig}\left(\left(X+iH\right)\Gamma\right).
\]
Consider the case where $W=X+iH$ is unitary, in this class. This
means $U^{\sigma}=U$ where $\sigma$ denotes conjugation by $\Gamma$.
This symmetry will hold also for $f$(U) so long as $f(\overline{\lambda})=f(\lambda).$
We can select a vertical line in the complex plane, near the imaginary
axis, that misses the spectrum of $U$ and get a homotopy $f_{t}$
from the identity function to the function that maps all to the right
of the line to $1$ and all to the left to $-1$. Thus we have a homotopy
to a unitary that is symmetric. Back in the $X$ and $H$ picture,
we can assume $H=0$. Along with the fact that $X$ is even, we are
in the situation
\[
X=\left[\begin{array}{cc}
A & 0\\
0 & B
\end{array}\right]
\]
where $A$ and $B$ are Hermitian. Since $H=0$ the index is
\[
\tfrac{1}{2}\mathrm{Sig}\left(A\right)-\tfrac{1}{2}\mathrm{Sig}\left(B\right).
\]
The reason this is the correct index to classify, where it seems we
need two indices, is that we can bring back nonzero H and use paths such
as
\[
U_{t}=\left[\begin{array}{cccccc}
\cos(t) &  &  & \sin(t)\\
 & 1 &  &  & 0\\
 &  & \ddots &  &  & \ddots\\
\sin(t) &  &  & -\cos(t)\\
 & 0 &  &  & -1\\
 &  & \ddots &  &  & \ddots
\end{array}\right]
\]
to increase the signature of $B$ by to while decreasing the signature
of $A$ by the same amount.

\subsection{Class BDI in 1D}

We can assume $H$ is odd and real and $X$ is even and real, where
even and odd is determined by the grading operator 
\[
\Gamma=\left[\begin{array}{cc}
I & 0\\
0 & -I
\end{array}\right].
\]
We can use the index from \S \ref{sub:Class-AIII-1D}, but that turns
out to be slower to compute that is necessary.

\begin{example}
Suppose $\alpha$ is real, between $-1$ and $1$, and set $\beta=\sqrt{1-\alpha^{2}}$.
Our example has
\[
X=\left[\begin{array}{cc}
\alpha & 0\\
0 & \alpha
\end{array}\right]
\]
 and
\[
H=\left[\begin{array}{cc}
0 & \beta\\
\beta & 0
\end{array}\right].
\]
These have the correct symmetries, and
\[
\left(X+iH\right)\Gamma=\left[\begin{array}{cc}
\alpha & -i\beta\\
i\beta & \alpha
\end{array}\right].
\]
This has signature $0$. It is unitarily equivalent to a real matrix,
which is not an accident.
\end{example}

We use the unitary $Q=\omega I+\overline{\omega}\Gamma$ 
where $\omega=\tfrac{1}{2}-\tfrac{i}{2}$.  We can conjugate by a unitary
matrix without altering the signature, and we compute 
\begin{align*}
Q\left(X+iH\right)\Gamma Q^{*} 
& =\left(\omega I+\overline{\omega}\Gamma\right)\left(X+iH\right)\Gamma\left(\overline{\omega}I+\omega\Gamma\right)\\
& =X\Gamma+H
\end{align*}
and discover a nice index,
\[
\frac{1}{2}\mathrm{Sig}\left(X\Gamma+H\right)
\]
which involves now the signature of a real symmetric matrix. See Remark~\ref{rem:complexSignature}
on computing signature, especially for sparse matrices.

The proof here that this invariant suffices is almost identical to
the argument in Section~\ref{sub:Class-AIII-1D}.

\subsection{Class D in 1D}

Since $\mathcal{C}^{2}=I$ we can assume $H^{\mathrm{T}}=-H$ and
$X^{\mathrm{T}}=X$. Since these are Hermitian, we see $X$ is real
and $H$ is pure-imaginary, so $X+iH$ is real. We are assuming it
to be invertible. The index here is 
\[
\mathrm{sign}\left(\det\left(X+iH\right)\right).
\]
We know the sign of the determinant classifies real orthogonal matrices
up to homotopy. We apply this to $X+iH$ and extract the needed Hermitian
and anti-Hermitian parts.

\subsection{Class DIII, 1D}

We can assume $H$ is imaginary and self-dual, and $X$ is real and
self-dual, but in this case we find that even and odd are to be determined
by the grading operator 
\[
\Gamma=\left[\begin{array}{cc}
0 & -iI\\
iI & 0
\end{array}\right].
\]
This ensures that the transpose is conjugated to the dual. We use
the unitary $Q=\omega I+\overline{\omega}\Gamma$ where 
$\omega=\tfrac{1}{2}-\tfrac{i}{2}$.  We compute 
\begin{align*}
Q\left(X+iH\right)\Gamma Q^{*} 
& =\left(\omega I+\overline{\omega}\Gamma\right)\left(X+iH\right)\Gamma\left(\overline{\omega}I+\omega\Gamma^{*}\right)\\
& =\left(\omega I+\overline{\omega}\Gamma\right)\left(X+iH\right)\left(\overline{\omega}\Gamma+\omega I\right)\\
& =-\frac{i}{2}\left(X+iH\right)+\frac{1}{2}\Gamma\left(X+iH\right)+\frac{1}{2}\left(X+iH\right)\Gamma+\frac{i}{2}\Gamma\left(X-iH\right)\Gamma\\
& =X\Gamma+H
\end{align*}
and we check its symmetries 
\[
\left(X\Gamma+H\right)^{*}=(\Gamma)X+H=\left(X\Gamma+H\right)
\]
and since $\Gamma$ and $H$ are imaginary and $X$ real, this is
imaginary, and so skew-symmetric. We can compute a Pfaffian, as before,
and our index is
\[
\mathrm{sign}\left(\mathrm{Pf}\left(i\left(X\Gamma+H\right)\right)\right).
\]

Here is a sketch of why this invariant classifies. As in class AIII
we are able to use functional calculus to reduce first to the case
of $X+iH$ being unitary and then also to where $H=0$. If X is real
symmetric and self-dual and unitary, it can be shown that it will
factor as 
\[
Q^{*}\left[\begin{array}{cc}
D\\
 & D
\end{array}\right]Q
\]
with $Q$ real orthogonal and symplectic and 
\[
D=\left[\begin{array}{cccccc}
1\\
 & \ddots\\
 &  & 1\\
 &  &  & -1\\
 &  &  &  & \ddots\\
 &  &  &  &  & -1
\end{array}\right].
\]
To see how to finish the classification, notice the invariant comes
out differently on
\[
X_{+}=\left[\begin{array}{cc}
1\\
 & 1
\end{array}\right],\quad X_{+}=\left[\begin{array}{cc}
-1\\
 & -1
\end{array}\right]
\]
since 
\[
\mathrm{Pf}\left(iX_{\pm}\Gamma\right)=\mathrm{Pf}\left(\left[\begin{array}{cc}
0 & \pm1\\
\mp1 & 0
\end{array}\right]\right)=\pm1.
\]
On the other hand
\[
\left[\begin{array}{cccc}
\cos(t) & 0 & 0 & \sin(t)\\
0 & \cos(t) & -\sin(t) & 0\\
0 & -\sin(t) & \cos(t) & 0\\
\sin(t) & 0 & 0 & \cos(t)
\end{array}\right]
\]
is a real, self-dual, Hermitian unitary path from $I_{4}$ to $-I_{4}$.

\subsection{Class CII, 1D}

This is much like the BDI situation. On \textbf{$\mathbf{M}_{4n}(\mathbb{C})$
}the operation corresponding to particle-charge conjugation is
\[
\left[\begin{array}{cc}
A & B\\
C & D
\end{array}\right]^{\kappa}=\left[\begin{array}{cc}
A^{\sharp} & C^{\sharp}\\
B^{\sharp} & D^{\sharp}
\end{array}\right].
\]
The grading operator is 
\[
\Gamma=\left[\begin{array}{cc}
I & 0\\
0 & -I
\end{array}\right].
\]
The operations corresponding to time-reversal is 
\begin{align*}
\left[\begin{array}{cc}
A & B\\
C & D
\end{array}\right]^{\tau} & =\left[\begin{array}{cc}
I & 0\\
0 & -I
\end{array}\right]\left[\begin{array}{cc}
A & B\\
C & D
\end{array}\right]^{\kappa}\left[\begin{array}{cc}
I & 0\\
0 & -I
\end{array}\right]\\
 & =\left[\begin{array}{cc}
I & 0\\
0 & -I
\end{array}\right]\left[\begin{array}{cc}
A^{\sharp} & C^{\sharp}\\
B^{\sharp} & D^{\sharp}
\end{array}\right]\left[\begin{array}{cc}
I & 0\\
0 & -I
\end{array}\right]\\
 & =\left[\begin{array}{cc}
A^{\sharp} & -C^{\sharp}\\
-B^{\sharp} & D^{\sharp}
\end{array}\right].
\end{align*}
To check these are the correct type, we notice
\[
\left[\begin{array}{cc}
A & B\\
C & D
\end{array}\right]^{\kappa}=-\left[\begin{array}{cc}
Z & 0\\
0 & Z
\end{array}\right]\left[\begin{array}{cc}
A & B\\
C & D
\end{array}\right]^{\mathrm{T}}\left[\begin{array}{cc}
Z & 0\\
0 & Z
\end{array}\right]
\]
and 
\[
\left[\begin{array}{cc}
A & B\\
C & D
\end{array}\right]^{\tau}=-\left[\begin{array}{cc}
Z & 0\\
0 & -Z
\end{array}\right]\left[\begin{array}{cc}
A & B\\
C & D
\end{array}\right]^{\mathrm{T}}\left[\begin{array}{cc}
Z & 0\\
0 & -Z
\end{array}\right],
\]
and since 
\[
\left[\begin{array}{cc}
Z & 0\\
0 & Z
\end{array}\right]^{2}=\left[\begin{array}{cc}
Z & 0\\
0 & -Z
\end{array}\right]^{2}=-\left[\begin{array}{cc}
I & 0\\
0 & I
\end{array}\right]
\]
we have $\mathcal{C}\circ\mathcal{C}=-I$ and $\mathcal{T}\circ\mathcal{T}=-I$.

Our matrices are Hermitian $X$ and $H$ with $X$ even and $X^{\kappa}=X$,
and with $H$ odd and $H^{\kappa}=H$. The index we use is
\[
\frac{1}{4}\mathrm{Sig}\left(X\Gamma+H\right).
\]
The reason for the extra factor of one-half will evident from the
symmetries here. We note
\[
\left(X\Gamma+H\right)^{*}=\Gamma X+H=X\Gamma+H
\]
and
\[
\left(X\Gamma+H\right)^{\tau}=\Gamma^{\tau}X+H=\Gamma X+H=X\Gamma+H
\]
so this matrix is ``self-dual'' and Hermitian, so has Kramer's doubling.

Arguing as before, we reduce to the case $H=0$ and $X$ unitary.
This means 
\[
X=\left[\begin{array}{cc}
Y\\
 & Z
\end{array}\right]
\]
with $Y$ and $Z$ both self-dual, Hermitian and unitary. Since
\[
\frac{1}{4}\mathrm{Sig}\left(X\Gamma\right)=\frac{1}{4}\mathrm{Sig}\left[\begin{array}{cc}
Y\\
 & -Z
\end{array}\right]=\frac{1}{4}\mathrm{Sig}\left(Y\right)-\frac{1}{4}\mathrm{Sig}\left(Z\right)
\]
we initially seem to be short by one invariant. The path
\[
U_{t}=\left[\begin{array}{cccc}
\cos(t) &  & \sin(t)\\
 & \cos(t) &  & \sin(t)\\
\sin(t) &  & -\cos(t)\\
 & \sin(t) &  & -\cos(t)
\end{array}\right]
\]
illustrates how to use nonzero $H_{t}$ to increase one signature
by four while decreasing the other by four.

\section{Dimension two, index formulas\label{sec:Dimension-two}}

What we are after here are invariants that, for two dimensional finite
systems on a square, can be quickly computed and that can explain
the robustness of gapless edge modes in the face of disorder. We only
consider disorder that respects the symmetry class. However, the simplicity
of these invariants should mean they function well for disorder that
is only approximately invariant under the needed symmetries.

\subsection{Class A in 2D}

We have no reality condition and use the index from Section~\ref{sec:Introduction},
at the origin, so
\[
\frac{1}{2}\mathrm{Sig}\left(X\otimes\sigma_{x}+Y\otimes\sigma_{y}+H\otimes\sigma_{z}\right).
\]

Notice that if consider
\[
H'=X\otimes\sigma_{x}+Y\otimes\sigma_{y}+H\otimes\sigma_{z}
\]
as a derived Hamiltonian, it constitutes a class A system in 0D.

\subsection{Class D in 2D\label{sub:Class-D-in-2D}}

We can assume, after perhaps a unitary change of basis, that $H$
is imaginary while $X$ and $Y$ are real, all being Hermitian. We
select our $\Gamma_{j}$ so that $H$ is tensored with $\sigma_{y}$,
which is imaginary. Therefore
\[
H'=X\otimes\sigma_{z}+Y\otimes\sigma_{x}+H\otimes\sigma_{y}
=
\left[\begin{array}{cc}
X & Y-iH\\
Y+iH & X
\end{array}\right]
\]
defines a 0D system in class AI. That is, it is real symmetric and,
by the local gap assumption, invertible. In terms of Table~\ref{tab:range-of-invariants}
this moves us two steps up and two to the left. This is similar to
the scattering matrix approach of Fulga et al.\ \cite{fulga2011scattering}
which, in slightly different geometry, moves one step up and one step
left. 

We utilize the invariant from Class AI in 0D and use
\[
\frac{1}{2}\mathrm{Sig}\left(X\otimes\sigma_{z}+Y\otimes\sigma_{x}+H\otimes\sigma_{y}\right).
\]

This is invariant under homotopy, is additive with respect to direct
sums, and is trivial on trivial systems. It could be the trivial invariant.
Perhaps the best way to show these invariants nontrivial is to use
them in an numerical study. We do that in some cases. Here we derive
the existence of a nontrivial example mathematically.

The standard example \cite{ChoiAlmostNotNearly,HastingsLoringWannier}
of three almost commuting Hermitian matrices
\[
\frac{1}{2}\mathrm{Sig}\left[\begin{array}{cc}
X & Y-iH\\
Y+iH & X
\end{array}\right]=\pm1
\]
has one matrix, say $H$, imaginary and the others real. All we are
missing is having $X$ and $Y$ commute. By the class AI version of
Lin's Theorem \cite{LorSorensenDisk} we can modify these a little
to produce $X_{1}$ and $Y_{1}$ that are commuting real orthogonal
matrices. If we use a large enough matrix size, $B(X,Y,H)$ will have
a large spectral gap, meaning the index will be unchanged when applied
to $(X_{1},Y_{1},H)$.

\subsection{Class DIII in 2D}

We are given symmetries $H^{\sharp}=H$, $X^{\sharp}=X$ and $Y^{\sharp}=Y$on
top of knowing $X$ and $Y$ are real and $H$ is imaginary, as in
Section~\ref{sub:Class-D-in-2D}. As we did there, we set
\[
H'=X\otimes\sigma_{z}+Y\otimes\sigma_{x}+H\otimes\sigma_{y}.
\]
This is real, so $H'^{\mathrm{T}\otimes\mathrm{T}}=H'$, while $H'^{\sharp\otimes\sharp}=H'.$
With these symmetry operations, the grading operator is $Z\otimes Z$.
We need to select partial isometries into the $1$ and $-1$ eigenspaces
for this grading operator, and choose
\[
W_{+}=\frac{1}{\sqrt{2}}\left[\begin{array}{cc}
I & 0\\
0 & I\\
0 & -I\\
I & 0
\end{array}\right],\quad W_{-}=\frac{1}{\sqrt{2}}\left[\begin{array}{cc}
I & 0\\
0 & I\\
0 & I\\
-I & 0
\end{array}\right].
\]
Now we can use the index from Class BDI in dimension zero and define
our index as
\[
\mathrm{sign}\left(\det\left(W_{+}^{*}H'W_{-}\right)\right).
\]
We want to see this is not a trivial index. 

Suppose $(X,Y,H)$ is any example from class D, which can exist with both odd
index and even. Let our class DIII example be the $2n$-by-$2n$ matrices
\begin{equation}
X_{1}=\left[\begin{array}{cc}
X\\
 & X
\end{array}\right],\quad Y_{1}=\left[\begin{array}{cc}
Y\\
 & Y
\end{array}\right],\quad H_{1}=\left[\begin{array}{cc}
H\\
 & -H
\end{array}\right].\label{eq:double_example}
\end{equation}
The index is then the sign of the determinant of
\begin{align*}
 & \frac{1}{2}\left[\begin{array}{cccc}
I & 0 & 0 & I\\
0 & I & -I & 0
\end{array}\right]\left[\begin{array}{cccc}
X & 0 & Y-iH & 0\\
0 & X & 0 & Y+iH\\
Y+iH & 0 & -X & 0\\
0 & Y-iH & 0 & -X
\end{array}\right]\left[\begin{array}{cc}
I & 0\\
0 & I\\
0 & I\\
-I & 0
\end{array}\right]\\
 & =\quad\left[\begin{array}{cc}
X & Y-iH\\
-Y-iH & X
\end{array}\right]\\
 & =\quad\left[\begin{array}{cc}
I & 0\\
0 & -I
\end{array}\right]\left[\begin{array}{cc}
X & Y-iH\\
Y+iH & -X
\end{array}\right]
\end{align*}
so the 
\[
\mathrm{ind}(X_{1},Y_{1},H_{1})=(-1)^{\mathrm{ind}(X,Y,H)}.
\]

\begin{rem}
Notice that in this example, the index comes out $1$ when we start
with $X$ and $Y$ and $H$ all commuting. An easy homotopy argument
shows we get trivial index of $1$ whenever we start in class DIII
with all three matrices commuting. This means we made valid choices
for $W_{\pm}$.
\end{rem}

\subsection{Class AII in 2D}

We have one symmetry $H^{\sharp}=H$, $X^{\sharp}=X$ and $Y^{\sharp}=Y$.
We set
\[
H'=X\otimes\sigma_{x}+Y\otimes\sigma_{y}+H\otimes\sigma_{z}
\]
and find $H'^{\sharp\otimes\sharp}=-H'.$ We are in a nonstandard
version of class D, dimension zero. The unitary matrix
\[
Q=\frac{1}{\sqrt{2}}\left[\begin{array}{cc}
I & -iZ\\
iZ & I
\end{array}\right]=\frac{1}{\sqrt{2}}\left[\begin{array}{cccc}
I & 0 & 0 & -iI\\
0 & I & iI & 0\\
0 & iI & I & 0\\
-iI & 0 & 0 & I
\end{array}\right]
\]
brings us to the standard picture and our invariant is
\[
\mathrm{sign}\left(\mathrm{Pf}\left(iQ^{*}H'Q\right)\right).
\]
We conduct a numerical study that provides ample evidence that this
is not a trivial invariant.

\begin{rem}
There are many choices here, including conventions as to the definition
of the Pfaffian, so that it is easy to program this wrong. Done correctly,
the index is $1$ when $H$ and $X$ and $Y$ all commute.
\end{rem}

\subsection{Class C in 2D}

We have one symmetry $H^{\sharp}=-H$, $X^{\sharp}=X$ and $Y^{\sharp}=Y$.
We set
\[
H'=X\otimes\sigma_{z}+Y\otimes\sigma_{x}+H\otimes\sigma_{y}
\]
so that $H'^{\sharp\otimes\mathrm{T}}=H'.$ We are in a nonstandard
version of class AII, dimension zero. The matrix 
\[
Q=\left[\begin{array}{cccc}
I & 0 & 0 & 0\\
0 & 0 & I & 0\\
0 & I & 0 & 0\\
0 & 0 & 0 & I
\end{array}\right]
\]
converts is to the standard dual operation. Our invariant is
\[
\frac{1}{4}\mathrm{Sig}\left(QH'Q\right)=\frac{1}{4}\mathrm{Sig}\left(H'\right).
\]
To show nontrivial values of this index are possible, consider any
example from class D. Let our class C example be
\begin{equation}
X_{1}=\left[\begin{array}{cc}
X\\
 & X
\end{array}\right],\quad Y_{1}=\left[\begin{array}{cc}
Y\\
 & Y
\end{array}\right],\quad H_{1}=\left[\begin{array}{cc}
H\\
 & H
\end{array}\right].\label{eq:double_example-C}
\end{equation}
Then
\begin{align*}
QH'Q & =Q\left[\begin{array}{cccc}
X & 0 & Y-iH & 0\\
0 & X & 0 & Y-iH\\
Y+iH & 0 & -X & 0\\
0 & Y+iH & 0 & -X
\end{array}\right]Q\\
 & =\left[\begin{array}{cccc}
X & Y-iH\\
Y+iH & -X\\
 &  & X & Y-iH\\
 &  & Y+iH & -X
\end{array}\right]\\
\end{align*}
so 
\[
\mathrm{ind}_{\mathrm{C}}(X_{1},Y_{1},H_{1})=\mathrm{ind}_{D}(X,Y,H).
\]

\section{Dimension three, index formulas\label{sec:Dimension-three}}

We believe in all five interesting classes in three dimensions we
can move up two and left two in Table~\ref{tab:The-symmetry-classes}
and arrive at a useful index. For now, we study this just in Class
AII. In this class, we have examples from physics research (with Hastings
\cite{HastLorTheoryPractice}) to show this invariant is interesting,
and theorems about the $K$-theory of real $C^{*}$-algebras (with
Boersema \cite{BL-realKO}) that allow is to identify the invariant.

The higher dimensional invariants in classes just one antiunitary
symmetry, can be explained using the theory of real (ungraded) $C^{*}$-algebras.
This will be discussed elsewhere.

\subsection{Class AII in 3D}

We have one symmetry $H^{\sharp}=H$,$X^{\sharp}=X$, $Y^{\sharp}=Y$,
and $Z^{\sharp}=Z$. We set
\[
H'=X\otimes\sigma_{x}+Y\otimes\sigma_{y}+Z\otimes\sigma_{z}
\]
and 
\[
X'=H\otimes I
\]
and find and find $H'^{\sharp\otimes\sharp}=-H'$ and $X'^{\sharp\otimes\sharp}=X'$
so we have a derived 1D system in class D. The unitary matrix
\[
Q=\frac{1}{\sqrt{2}}\left[\begin{array}{cccc}
I & 0 & 0 & -iI\\
0 & I & iI & 0\\
0 & iI & I & 0\\
-iI & 0 & 0 & I
\end{array}\right]
\]
brings us to the standard picture and our invariant is
\[
\mathrm{sign}\left(\det\left(Q^{*}\left(X'+iH'\right)Q\right)\right)=\mathrm{sign}\left(\det\left(X'+iH'\right)\right).
\]
 Again we offer a numerical study that provides ample evidence that
this is not the trivial invariant. 
\begin{rem}
The advantage of the left form of the invariant is that an LU factorization
will be faster and take less memory.
\end{rem}

\section{Bulk-edge correspondence}

We now prove a relation between the index values in the
pseudoresolvant and the pseudospecturm.  Essentially, between
two points in energy-position space where the index changes,
there must be a point in the pseudospectrum.  So at that
point there must be a vector approximately localized in
position and energy.  For weak disorder, this will mean
a vector localized at the edge and localized at the Fermi
level.  For strong disoroder, things get messier, with
the ``edge modes'' forming a ring around the sample, sometimes
moved in from the edge.  

Unlike in \cite{EsinGurari-Bulk-boundary} and \cite{Kane_et_al_Zmod2QuantumHall},
for example, we do not have a separate Hamiltonian for the edge states.
Rather we are classifying the approximate zero modes of the Hamiltonian.
In a system with weak disorder and nontrivial invariant in the center, these
approximate zero modes cannot localize in the bulk, but they can and do
localize near the edge.  Since the value of the invariant is robust against 
disorder, at least disorder with the correct symmetry, whatever is happening
at the edge is robust.

We are not claiming the converse.  Most likely, in higher 
dimensions there are  robust edge effects that can exist without
these particular invariants being nonzero.  Indeed, see the
discussion of stabilization in $K$-theory in \cite{kitaev-2009}.

We will use $\mathrm{ind}(X_{1},\dots,X_{D+1})$ generically for any
of the indices described on Sections~\ref{sec:Dimension-zero}-\ref{sec:Dimension-three}.
We do not assume in this section that the first $D$ matrices commute.
Our convention is to call $X_{D+1}$ the Hamiltonian, even if these
matrices are not related to quantum systems.

We will assume the needed symmetries on the $\Gamma_{j}$ to correspond
the the choices made in the definition of a specific index. For example
when $D=3$ and we are in Class AII, the index is
\begin{equation}
\mathrm{sign}\left(\det\left(Q^{*}\left[\begin{array}{cccc}
0 & 0 & H+iZ & iX+Y\\
0 & 0 & iX-Y & H-iZ\\
H-iZ & -iX-Y & 0 & 0\\
-iX+Y & H+iZ & 0 & 0
\end{array}\right]Q\right)\right)\label{eq:3D-AII-full formula}
\end{equation}
so our choices for the $\Gamma_1,\Gamma_2,\Gamma_3,\Gamma_4$ are
\[
\left[\begin{array}{cccc}
0 & 0 & 0 & i\\
0 & 0 & i & 0\\
0 & -i & 0 & 0\\
-i & 0 & 0 & 0
\end{array}\right]
,\quad
\left[\begin{array}{cccc}
0 & 0 & 0 & 1\\
0 & 0 & -1 & 0\\
0 & -1 & 0 & 0\\
1 & 0 & 0 & 0
\end{array}\right]
,\quad
\left[\begin{array}{cccc}
0 & 0 & i & 0\\
0 & 0 & 0 & -i\\
-i & 0 & 0 & 0\\
0 & i & 0 & 0
\end{array}\right]
,\quad
\left[\begin{array}{cccc}
0 & 0 & 1 & 0\\
0 & 0 & 0 & 1\\
1 & 0 & 0 & 0\\
0 & 1 & 0 & 0
\end{array}\right].
\]
What is essential is that the index can
be computed continuously from
\[
B(X_{1},\dots,X_{D+1})
\]
whenever this
is invertible. This means that given a continuous path \[(X_{1,t},\dots,X_{D+1,t})\]
the only way for to change is for it to be undefined at some point
because \[B(X_{1,t},\dots,X_{D+1,t})\] is singular. 

\begin{defn}
Assume $(X_{1},\dots,X_{D+1})$ are in some symmetry class. If that
symmetry class does not have any symmetry that anticommute with the
Hamiltonian, then for any $(D+1)$-tuple $\boldsymbol{\lambda}$ not
in $\Lambda(X_{1},\dots,X_{D+1})$ define 
\[
\mathrm{ind}_{\boldsymbol{\lambda}}(X_{1},\dots,X_{D+1})
=
\mathrm{ind}\left(X_{1}-\lambda_{1}I,\dots,X_{D+1}-\lambda_{D+1}I\right).
\]
In the other classes, we do the same except only define this index
at a point $\boldsymbol{\lambda}$ when $\lambda_{D+1}=0$.
\end{defn}

\begin{lem}
\label{lem:UniformContinuityOfGap} If $(X_{1},\dots,X_{D+1})$ and
$(Y_{1},\dots,Y_{D+1})$ are tuples of Hermitian matrices of the same
size then
\[
\left|\left\Vert \left(B(X_{1},\dots,X_{D+1})\right)^{-1}\right\Vert ^{-1}-\left\Vert \left(B(Y,\dots,Y_{D+1})\right)^{-1}\right\Vert ^{-1}\right|
\leq
\sqrt{\sum\left\Vert X_{j}-Y_{j}\right\Vert ^{2}}.
\]
If $(X_{1},\dots,X_{D+1})$ and $(Y_{1},\dots,Y_{D+1})$ are in the
same symmetry class and 
\[
\mathrm{ind}(X_{1},\dots,X_{D+1})\neq\mathrm{ind}(Y,\dots,Y_{D+1})
\]
then
\[
\left|\left\Vert \left(B(X_{1},\dots,X_{D+1})\right)^{-1}\right\Vert ^{-1}+\left\Vert \left(B(Y,\dots,Y_{D+1})\right)^{-1}\right\Vert ^{-1}\right|
\leq
\sqrt{\sum\left\Vert X_{j}-Y_{j}\right\Vert ^{2}}
\]
and somewhere on the line segment $(X_{1,t},\dots,X_{D+1,t})$ between
these pairs is a point where
\[
B(X_{1,t},\dots,X_{D+1,t})
\]
is singular.
\end{lem}

\begin{proof}
We apply Weyl's estimate on the spectral variation of Hermitian matrices
to 
\[
B(X_{1,t},\dots,X_{D+1,t})
\]
 at various places along the line defined by
\[
X_{j,t}=(1-t)X_{j}+Y_{j}.
\]
Let's let $a$ be the eigenvalue of smallest magnitude for $B(X_{1},\dots,X_{D+1})$,
so
\[
a=\pm\left\Vert \left(B(X_{1},\dots,X_{D+1})\right)^{-1}\right\Vert ^{-1}
\]
and similarly have an eigenvalue $b$ for $B(Y,\dots,Y_{D+1})$ with
\[
b=\pm\left\Vert \left(B(Y_{1},\dots,Y_{D+1})\right)^{-1}\right\Vert ^{-1}.
\]
Without loss of generality, $\left|a\right|<\left|b\right|$. The
closest eigenvalue to $a$ in the spectrum of $B(Y_{1},\dots,Y_{D+1})$
is at least at a distance of $\left|b\right|-\left|a\right|$ away.
This is then a lower bound on the best overall spectral pairing, and
that in turn a lower bound on the distance between the matrices. 

If the index varies, then at some $t$ the matrix 
\[
B(X_{1,t},\dots,X_{D+1,t})
\]
is singular, since the index is continuous where defined. We can apply
the first statement in the result to $B(X_{1},\dots,X_{D+1})$ and
$B(X_{1,t},\dots,X_{D+1,t})$, and then again to $B(Y_{1},\dots,Y_{D+1})$
and $B(X_{1,t},\dots,X_{D+1,t})$. These points are collinear so the
estimates add.
\end{proof}

\begin{lem}
\label{lem:BdistanceIsEuclideanDistance.}For any two $(D+1)$-tuples
$\boldsymbol{\lambda}$ and $\boldsymbol{\mu}$ of scalars, 
\[
\left\Vert B\left(\lambda_{1}I,\dots,\lambda_{D+1}I\right)-B\left(\mu_{1}I,\dots,\mu_{D+1}I\right)\right\Vert 
=
d(\boldsymbol{\lambda},\boldsymbol{\mu})
\]
where $d$ denotes Euclidean distance.
\end{lem}

\begin{proof}
We first notice 
\[
B\left(\lambda_{1}I,\dots,\lambda_{D+1}I\right)-B\left(\mu_{1}I,\dots,\mu_{D+1}I\right)
=
B\left(\left(\lambda_{1}-\mu_{1}\right)I,\dots,\left(\lambda_{D+1}-\mu_{D+1}\right)I\right)
\]
so we need only compute 
\[
\left\Vert B\left(\mu_{1}I,\dots,\mu_{D+1}I\right)\right\Vert .
\]
By Equation~\ref{eq:eval_B_squared},
\[
B\left(\gamma_{1}I,\dots,\gamma_{D+1}I\right)^{2}=\sum\gamma_{j}^{2}I
\]
so, by the spectral mapping theorem, 
\[
\sigma\left(B\left(\gamma_{1}I,\dots,\gamma_{D+1}I\right)\right)
\subseteq
\left\{ \pm\sqrt{\sum\gamma_{j}^{2}}\right\} .
\]
Therefore 
\[
\left\Vert B\left(\gamma_{1}I,\dots,\gamma_{D+1}I\right)\right\Vert 
=
\sqrt{\sum\gamma_{j}^{2}}.
\]
\end{proof}

So now we know that points in the pseudospectrum with different index
must be separated by certain distance. We also have a means to estimate
how much disorder is needed to change an index. The reason nonzero
index is thus important, when it is found, is that trivial index prevails
away from the origin.

\begin{lem}
\label{lem:IndezZeroNearInfinty}If $(X_{1},\dots,X_{D+1})$ is in
any symmetry class, and if 
\[
\left|\boldsymbol{\lambda}\right|>\left\Vert B(X_{1},\dots,X_{D+1})\right\Vert 
\]
 then $\boldsymbol{\lambda}$ is in the pseudoresolvent and 
\[
\mathrm{ind}_{\boldsymbol{\lambda}}\left(X_{1},\dots,X_{D+1}\right)=0.
\]
\end{lem}

\begin{proof}
Notice first that $B\left(\lambda_{1}I,\dots,\lambda_{D+1}I\right)$
has spectrum $\{\pm\left|\boldsymbol{\lambda}\right|\}$ with trivial
index, in any symmetry class. We then consider the path 
$B_{\boldsymbol{\lambda}}(tX_{1},\dots,tX_{D+1})$
which has length
\[
\left\Vert B(X_{1},\dots,X_{D+1})\right\Vert .
\]
 For the signature along this path to change, it will need to be at
least as long as $\left|\boldsymbol{\lambda}\right|$.
\end{proof}

We now prove bulk-edge correspondence. To accommodate highly disordered
systems, we allow for a generous interpretation of edge mode. A nontrivial
index in the center will need to be surrounded, but not too closely,
by a ring of approximate modes approximately at the Fermi level (assumed
here to be zero). We can, for less disordered systems, compute the
index well away from the center and see that the edge modes are really
near the edge. 

\begin{thm}
\label{thm:Bulk-Edge}Suppose $(X_{1},\dots,X_{D},H)$ is in any symmetry
class and that
\[
\mathrm{ind}(X_{1},\dots,X_{D},H)
\]
is nontrivial and 
\[
R=\left\Vert \left(B(X_{1},\dots,X_{D},H)\right)^{-1}\right\Vert ^{-1}
\]
is positive. Then for every unit $D$-tuple $(\alpha_{1},\dots,\alpha_{D})$,
at some point along the ray 
\[
t\mapsto\boldsymbol{\lambda}_{t}=(t\alpha_{1},\dots,t\alpha_{D},0)
\]
there is a $\boldsymbol{\lambda}_{s}$ which is in the pseudospectrum.
If $\boldsymbol{\lambda}_{t}$ is in the pseudospectrum then 
\[
R
\leq
\left|\boldsymbol{\lambda}_{t}\right|
\leq
\sqrt{
\left\Vert H\right\Vert ^{2}+\sum_{j}\left\Vert X_{j}\right\Vert ^{2}+\sum_{j}\left\Vert [X_{j},H]\right\Vert .
}.
\]
\end{thm}

\begin{proof}
Suppose $\boldsymbol{\lambda}_{t}$ is in the pseudospectrum. The
claim $\left|\boldsymbol{\lambda}_{t}\right|\geq R$ follows directly
from Lemmas \ref{lem:UniformContinuityOfGap} and \ref{lem:BdistanceIsEuclideanDistance.}.
We know from Lemma~\ref{lem:IndezZeroNearInfinty} that $\boldsymbol{\lambda}_{t}$
is in the pseudoresolvent whenever 
\[
\left|\boldsymbol{\lambda}_{t}\right|>\left\Vert B(X_{1},\dots,X_{D},H)\right\Vert .
\]
We now use Equation~\ref{eq:eval_B_squared} to get a bound on this
norm,
\begin{align*}
\left\Vert B(X_{1},\dots,X_{D},H)\right\Vert ^{2} 
& =\left\Vert B(X_{1},\dots,X_{D},H) ^{2} \right\Vert \\
& \leq\left\Vert H^{2}+\sum X_{j}^{2}\right\Vert +\left\Vert \sum[X_{j},H]\right\Vert \\
& \leq\left\Vert H\right\Vert ^{2}+\sum\left\Vert X_{j}\right\Vert ^{2}+\sum\left\Vert [X_{j},H]\right\Vert .
\end{align*}
If $\boldsymbol{\lambda}_{t}$ is in the pseudospectrum
then
\[
\left|\boldsymbol{\lambda}_{t}\right|^2
\leq
\left\Vert H\right\Vert ^{2}+\sum\left\Vert X_{j}\right\Vert ^{2}+\sum\left\Vert [X_{j},H]\right\Vert .
\]
Again using Lemma~\ref{lem:IndezZeroNearInfinty} we see that somewhere
on this ray the index is trivial. At the start of the ray the index
was assumed to be nontrivial, so Lemma~\ref{lem:UniformContinuityOfGap}
tells us that at least one point on the ray is in the pseudospectrum.
\end{proof}

\begin{rem}
In the symmetry classes without particle-hole conjugation, we can
consider rays that slant up or down in the energy direction. In those
classes, the edge modes persist above and below the Fermi level. So
when $D=2$ in a class like AII, the pseudospectrum contains a sphere
that surrounds a hole. So starting with lattice geometry of a square,
we have produced a modified sphere and nontrivial topology.
\end{rem}

\section{Local and global invariants}

Suppose we have increasing system sizes, say $(X_{k},Y_{k},H_{k})$
with consistent spacial units, such as nanometers. We then must select
rescaling of units $\eta_{k}>0$ for all $k$ and so work with the
observables $(\eta X_{k},\eta Y_{k},H_{k})$. Another view point is
that we will be basing our pseudospectrum and index on
\[
B(\eta_{k}X_{k},\eta_{k}Y,H)-B(\eta_{k}\lambda_{1}I,\eta_{k}\lambda_{2}I,\lambda_{3}I).
\]

In terms of approximate modes, this means we are producing $\mathbf{v}$
that more or less minimize
\[
\max\left(
\eta_{k}\left\Vert X_{k}\mathbf{v}-\lambda_{1}\mathbf{v}\right\Vert ,
\eta_{k}\left\Vert Y_{k}\mathbf{v}-\lambda_{2}\mathbf{v}\right\Vert ,
\left\Vert H_{k}\mathbf{v}-\lambda_{3}\mathbf{v}\right\Vert 
\right)
\]
which, for points in the pseudospectrum, will be on the order of
\[
\sqrt{\eta_{k}}
\sqrt{
\left\Vert \left[H_{k},X_{k}\right]\right\Vert ,\left\Vert \left[H_{k},Y_{k}\right]\right\Vert 
}.
\]
I.e.
\[
\max\left(
\sqrt{\eta_{k}}\left\Vert X_{k}\mathbf{v}-\lambda_{1}\mathbf{v}\right\Vert ,
\sqrt{\eta_{k}}\left\Vert Y_{k}\mathbf{v}-\lambda_{2}\mathbf{v}\right\Vert ,
\frac{1}{\sqrt{\eta_{k}}}\left\Vert H_{k}\mathbf{v}-\lambda_{3}\mathbf{v}\right\Vert 
\right)
\approx C
\]
If we are trying to model a local probe, we would then want to keep
$\eta_{k}$ constant. Ideally it would be set to correspond to modes
localized in energy so most energy spectrum can fit in the band gap
of the clean periodic Hamiltonian and with spread in position of a
few nanometers. This will to correspond to the expected spacial resolution
of a scanning tunneling spectroscope. If we hold $\eta_{k}$ constant
we will have a \emph{local index} and local pseudospectrum.

If we instead think bulk-edge correspondence, as in Theorem~\ref{thm:Bulk-Edge},
then we want 
\[
R_{k}
=
\frac{1}{\eta_{k}}\left\Vert \left(B(\eta_{k}X_{k},\eta_{k}Y_{k},H_{k})\right)^{-1}\right\Vert ^{-1}
\]
to grow proportionate to $\left\Vert X_{k}\right\Vert $ which, for
simplicity, we take equal to $\left\Vert Y_{k}\right\Vert $. At the
very least, we want this to grow, so we need $\eta_{k}\searrow0$.
In keeping with how we rescale the periodic observables in the torus
geometry, we use $\eta_{k}=\eta_{0}/\left\Vert X_{k}\right\Vert $
and call the resulting index a \emph{global index}.

\section{Dimension two, numerics}

We present examples in two symmetry classes, AI and AII. We look at
examples illustrating the local invariant, and do a study of the disorder-averaged
global invariant, looking for the type of phase transition out of
topological insulator caused by doping. We use Hamiltonians previously
used in such studies so we can test the new algorithms.

We have collected ample evidence that the new indices correspond to
old indices, in three combinations of dimension and symmetry classes,
in the special case of modest disorder strength and where the Fermi
level is in the middle of the gap of the periodic Hamiltonian of clean
system. We also get equality for trivial systems, meaning systems
in the atomic limit. What can we expect to prove here?

We are discussing invariants that take discrete values on individual
finite systems. These are closely related to the $K$-theory obstructions
to fixing approximate matrix representations of $C^{*}$-algebra relations
to be exact representations \cite{EilersLoringContingenciesStableRelations}.
Based on that older research, and the more recent work \cite{LoringQuantKth},
we expect to soon find a rigorous proof that the new index formulas
agree with some established index in the case of weak disorder and
with the Fermi level in the middle of a big bulk gap. For the more
interesting situation, exploring transitions between topological and
ordinary insulating states, we expect at best a probablistic result.
Moreover different indices will mark the transitions differently.
We hope to have established a connection with these new indices and
edge modes, and the practicality of working with systems of nontrivial
size. It will take time to do numerical studies contrasting the phase transitions
as found by various index algorithms, including 
\cite{Mondr_Prodan_AIII_1D,prodan2011disordered}
and \cite{fulga2011scattering,fulga2011Wire}.

\subsection{Class AI in 2D}

\begin{figure}
\hspace*{-1.8in}(a)\includegraphics[clip,scale=0.4,bb=90 340 520 440]{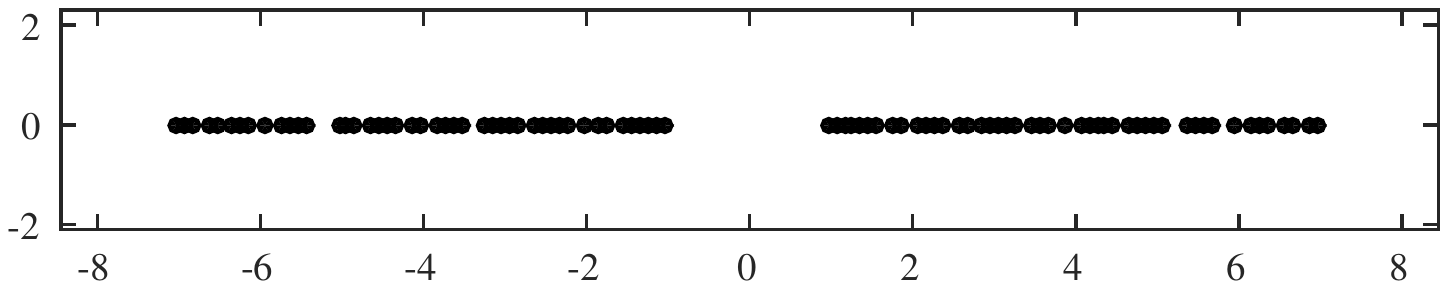}\\
\hspace*{-1.8in}(b)\includegraphics[clip,scale=0.4,bb=90 340 520 440]{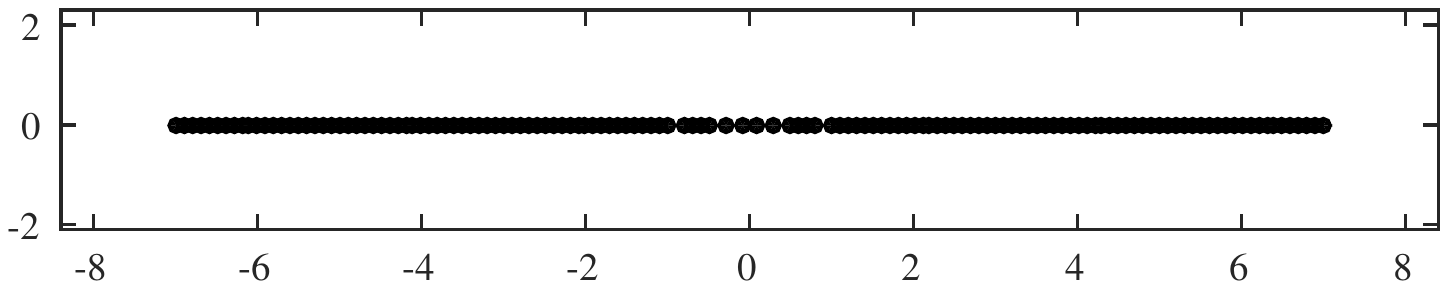}\\

\vspace*{-1.3in}
\hspace*{3.5in}
(c)\includegraphics[clip,scale=0.9,bb=210 300 410 490]{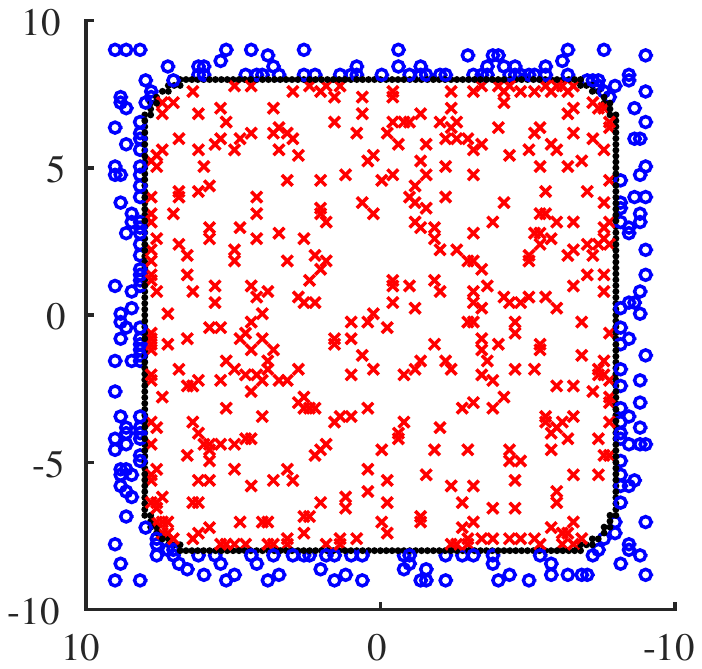}
\caption{Chern insulator with zero disorder, looking down at the pseudospectrum
and labled pseudoresolvant at the Fermi level. 
Panel (a) is the $0.1$-pseudospectrum of the Hamiltonian with periodic
boundary conditions. Panel (b) is the $0.1$-pseudospectrum of the
Hamiltonian with zero boundary conditions. Panel (c) is the Clifford
pseudospectrum with $K$-theory labels in the gaps, but only at the
Fermi level of zero.
\label{fig:AI-local-1}
}
\end{figure}

\begin{figure}
\hspace*{-1.8in}(a)\includegraphics[clip,scale=0.4,bb=90 340 520 440]{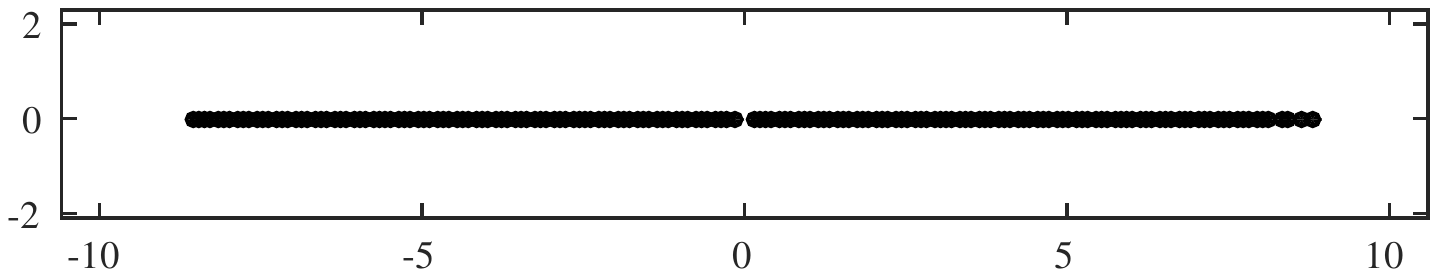}\\
\hspace*{-1.8in}(b)\includegraphics[clip,scale=0.4,bb=90 340 520 440]{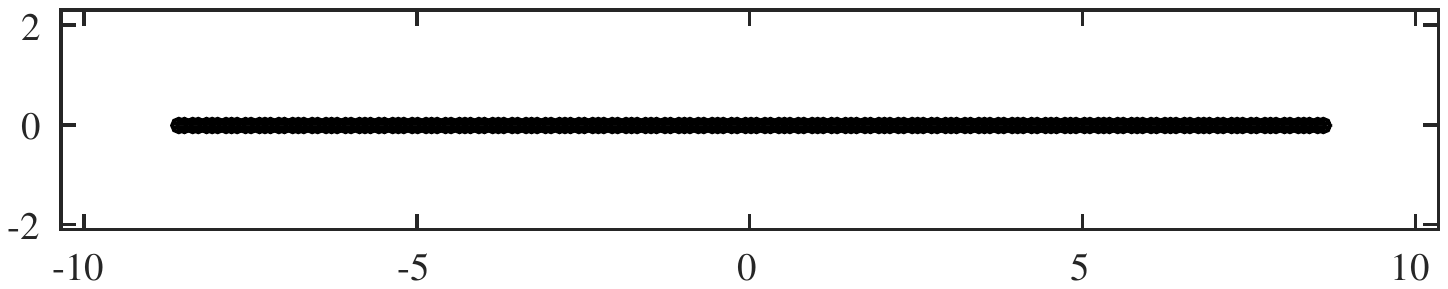}\\

\vspace*{-1.3in}
\hspace*{3.5in}
(c)\includegraphics[clip,scale=0.9,bb=210 300 410 490]{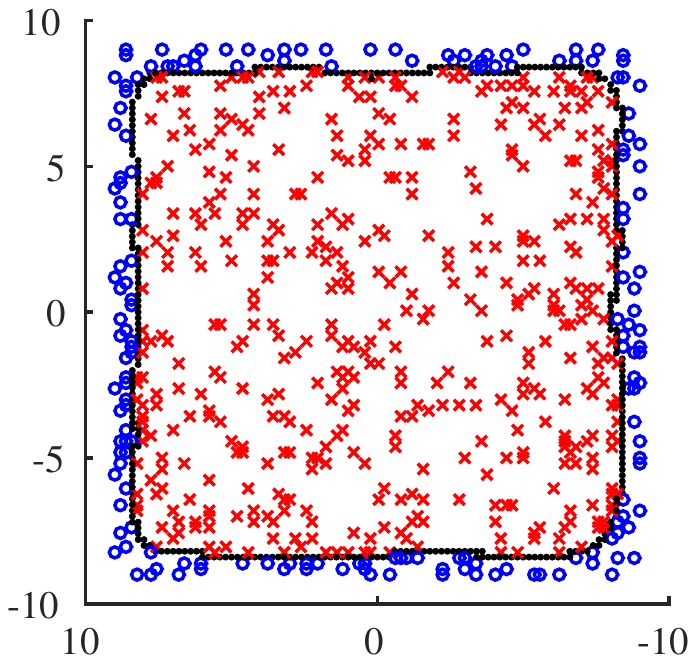}
\caption{Chern insulator with disorder set at $7$. 
Panel (a) is the $0.1$-pseudospectrum of the Hamiltonian with periodic
boundary conditions. Panel (b) is the $0.1$-pseudospectrum of the
Hamiltonian with zero boundary conditions. Panel (c) is the Clifford
pseudospectrum with $K$-theory labels in the gaps. The bulk gap is
just closing at this disorder strength. The small dots (black) are
in the pseudospectrum. Some of the grid points in the pseudoresolvant
are labeled by there index, with crosses (red) indicating index $1$
and with circles (blue) indicating index $-1$. 
\label{fig:AI-local-2}
}
\end{figure}

We look again at the model Chern insulator as in Example~\ref{exa:CherInsulator}.
We create the usual Hamiltonian $H_{\mathrm{per}}$ with periodic
boundary conditions as well as $H,$ the Hamiltonian with zero at
the square boundary. We look now at the $\epsilon$-pseudospectrum
only at energy zero and compute the local index, using $\eta=0.02$ and
$\epsilon=0.05$. For comparison we compute the spectrum of both $H$
and $H_{\mathrm{per}}$. In fact we plot the $\epsilon$-pseudospectrum
for small $\epsilon$ as this is known to close to the actual spectrum
and is much faster to compute. 

\begin{figure}
\includegraphics[clip,scale=0.9,bb=210 300 410 490]{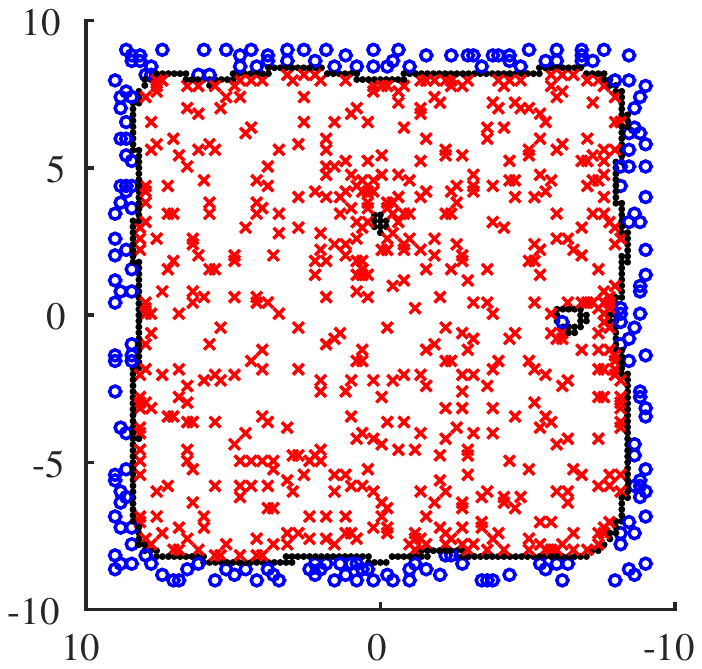}
\includegraphics[clip,scale=0.9,bb=210 300 410 490]{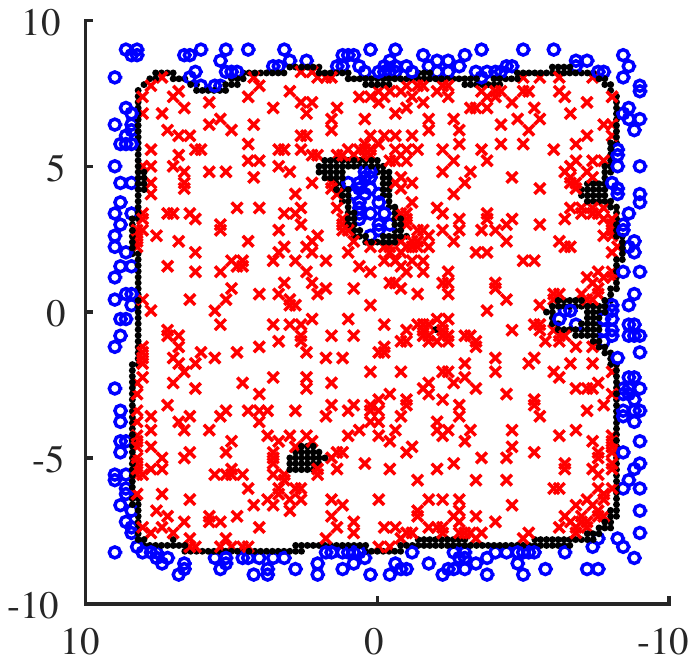}
\caption{Chern insulator with disorder set to 8 (left) or 9 (right). 
\label{fig:AI-local-3}
}
\end{figure}

\begin{figure}
\includegraphics[clip,scale=0.9,bb=210 300 410 490]{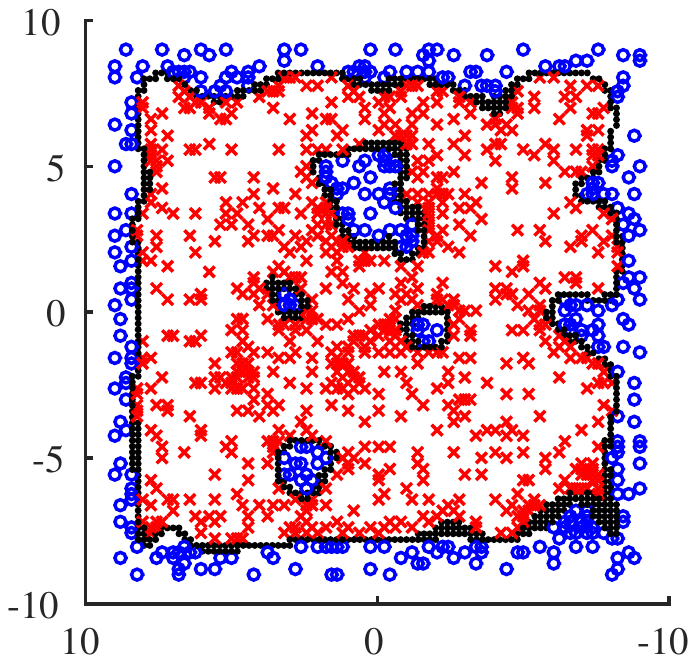}
\includegraphics[clip,scale=0.9,bb=210 300 410 490]{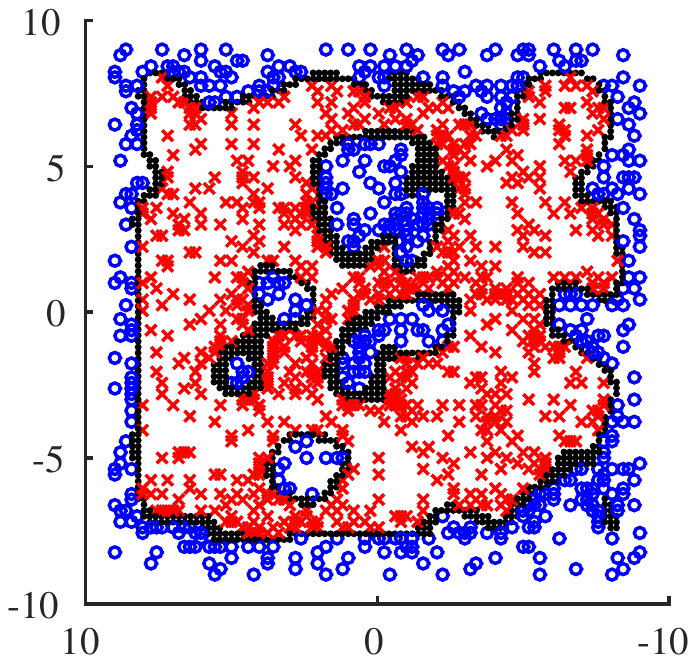}
\caption{Chern insulator with disorder set to 10 (left) or 11 (right).
\label{fig:AI-local-4} 
}
\end{figure}

\begin{figure}
\includegraphics[clip,scale=0.9,bb=210 300 410 490]{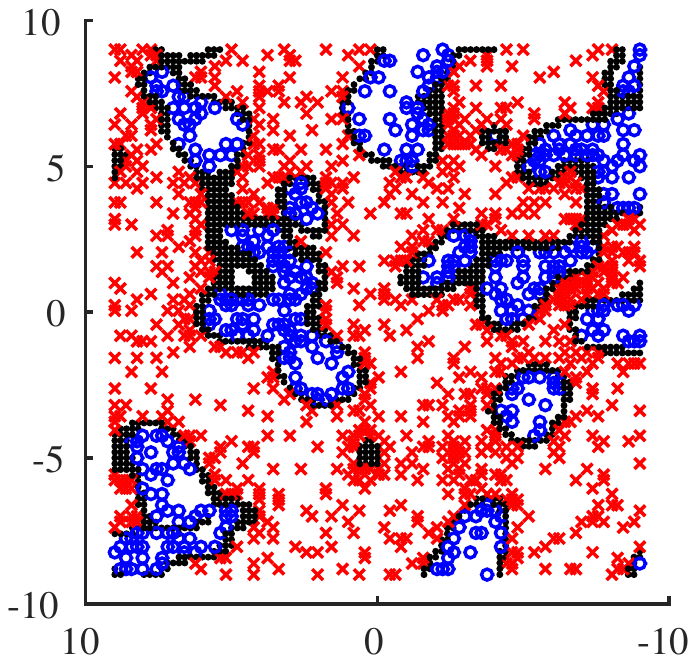}
\includegraphics[clip,scale=0.9,bb=210 300 410 490]{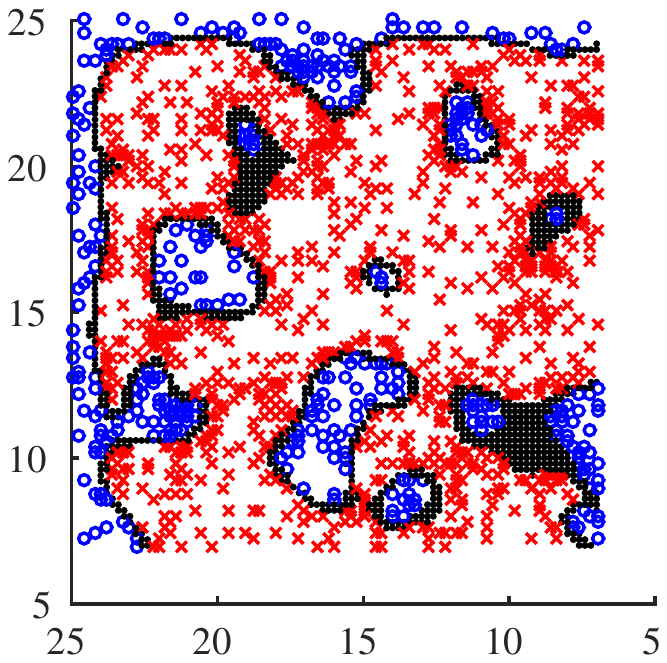}
\caption{Chern insulator with disorder set to 11. 
Now on an 50-by-50 lattice, just in the middle of the bulk (left)
and at the top-right corner (right).
\label{fig:AI-local-5}  
}
\end{figure}

We look at an $18$-by-$18$ lattice with increasing disorder, in Figures
\ref{fig:AI-local-1}-\ref{fig:AI-local-4}. At
the high value of disorder substantially, as in Figure
\ref{fig:AI-local-4}, the zero modes have, in places, moved in substantially
from the edge. On a $50$-by-$50$,
in Figure~\ref{fig:AI-local-5}, we see better how the center of
the sample has roughly circular patches of the wrong index, while
the corner of the sample still has some semblance of a boundary effect.

\begin{figure}[pt]
\includegraphics[bb=1in 3in 7.5in 8in,clip,scale=0.45]{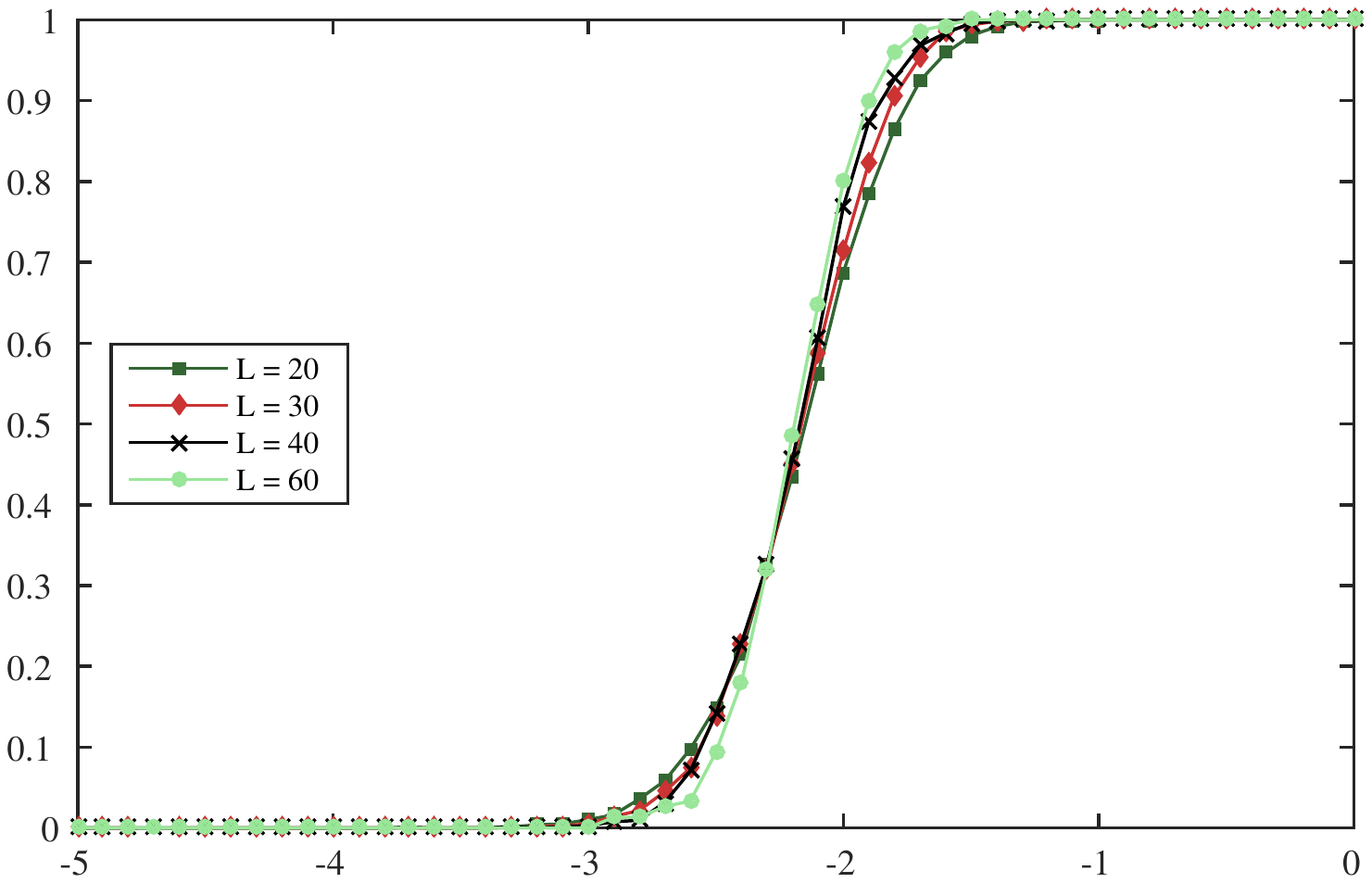}
\includegraphics[bb=1in 3in 7.5in 8in,clip,scale=0.45]{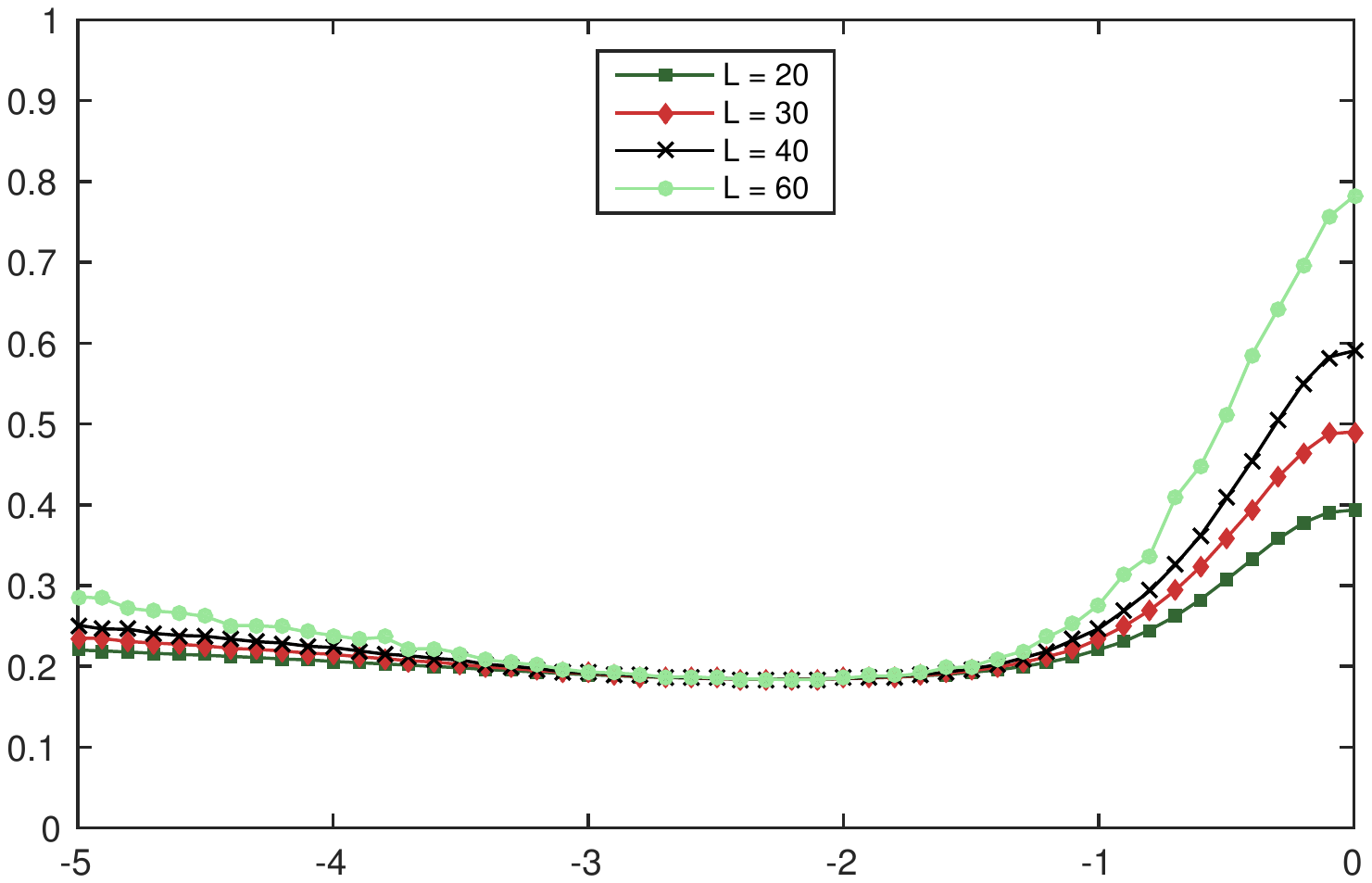}
\caption{Chern Insulator by the old Bott index. The number $N$ of samples
for an $L$-by-$L$ lattice used in these averages was: $L=20$, $N=10938$;
$L=30$, $N=2634$; $L=40$, $N=104$0; $L=60$, $N=150$. The left
panel shows the disorder-averaged Bott index with various Fermi levels.
The right panel shows the disorder-average of the concentration of
the Fermi projector, as defined in Equation~\ref{eq:concentrationDef}.
\label{fig:Chern-average-old}
}
\end{figure}

Now we look at the global invariant with $\eta=4/L$ when the model
is on an $L$-by-$L$ lattice. Holding disorder fixed at $8$ we compute
the index at the center, but with the energy level moving between
$-5$ and $0$. The results are shown in Figure~\ref{fig:Chern-average-new}

\begin{figure}[pt]
\includegraphics[bb=1in 3in 7.5in 8in,clip,scale=0.45]{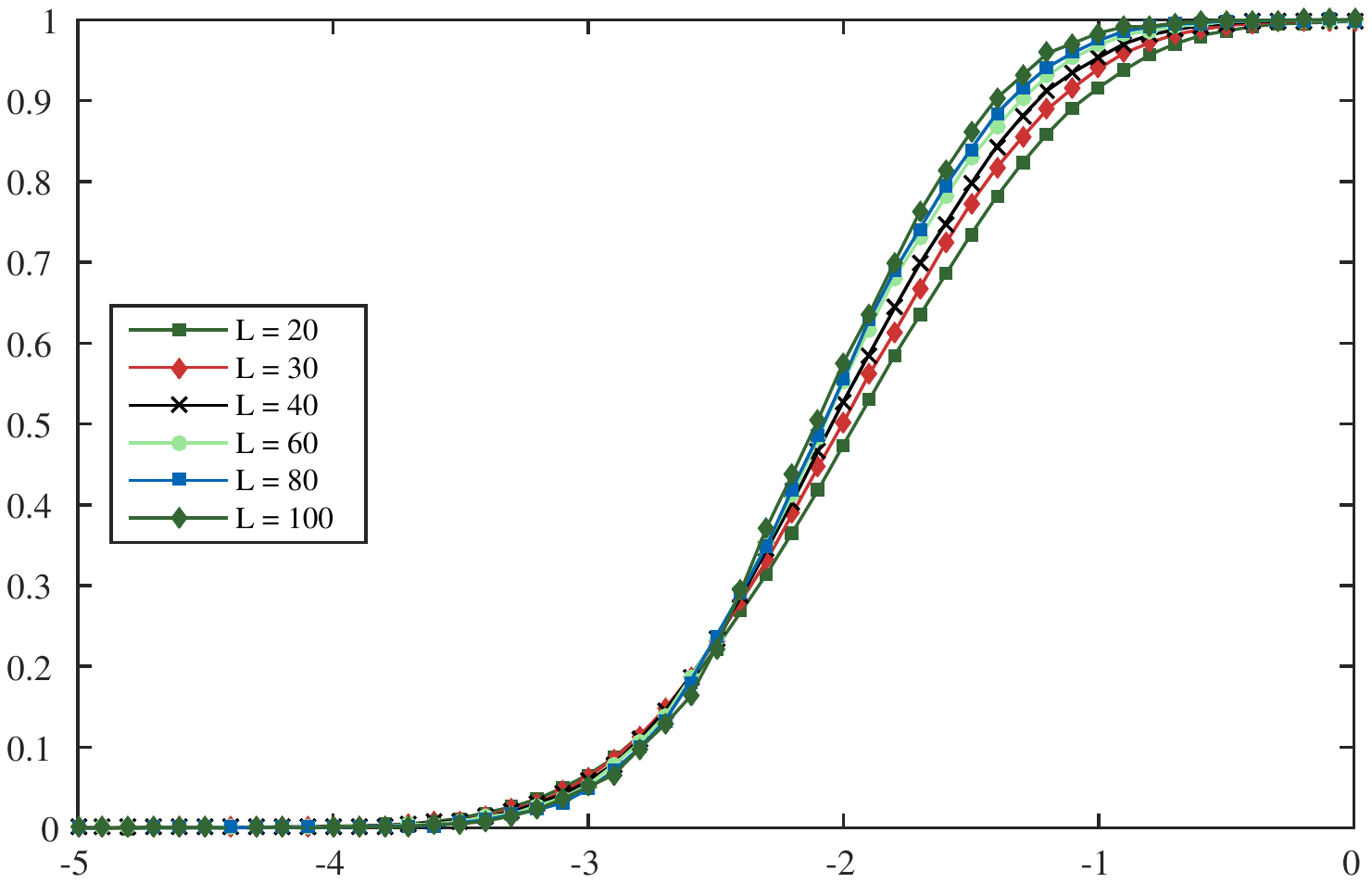}
\includegraphics[bb=1in 3in 7.5in 8in,clip,scale=0.45]{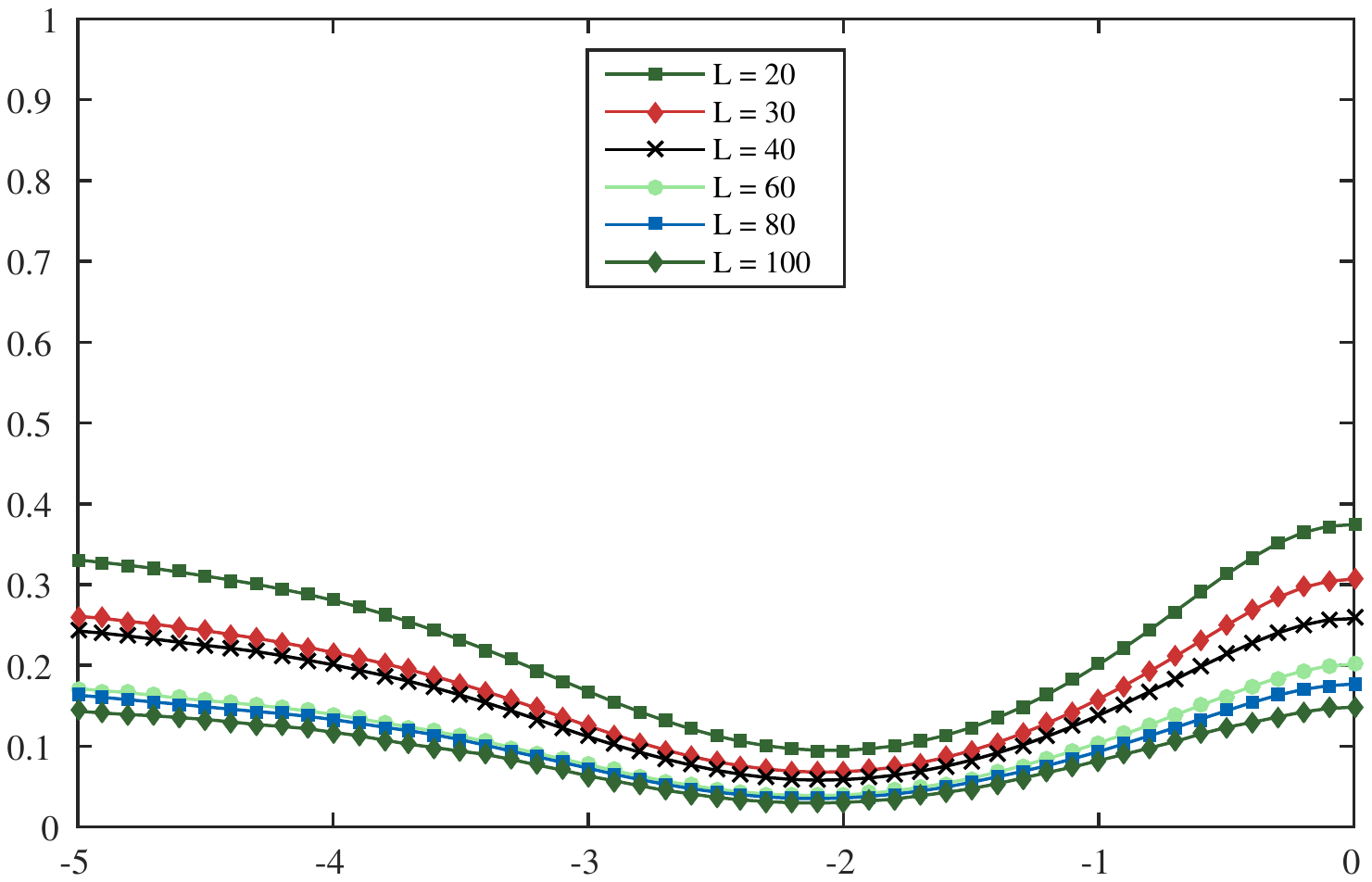}
\caption{Chern Insulator by the new index. The number $N$ of samples for an
$L$-by-$L$ lattice used in these averages was: $L=20$, $N=15929$;
$L=30$, $N=21326$; $L=40$, $N=16853$; $L=60$, $N=10224$; $L=80$,
$N=865$9; $L=100$, $N=2845$. The left panel shows the disorder-averaged
new index with various Fermi levels. The right panel shows the disorder-average
of gap localized at the origin, as defined in Equation~\ref{eq:GapAtOrigin}.
\label{fig:Chern-average-new}
}
\end{figure}

We did this study with dense matrix methods and the formula for the
Bott index, in joint work with Hastings \cite{LorHastHgTe}. We re-ran
this study in order to compute the following proxy for the inverse
of spread of the Fermi projector $P=P_{E_{F}}$,
\begin{equation}
\mathrm{concentration}\left(P_{E_{F}}\right)
=
1/\sqrt{\mbox{\ensuremath{\left\Vert \eta\left[\hat{U},P\right]\right\Vert ^{2}}+\ensuremath{\left\Vert \eta\left[\hat{V},P\right]\right\Vert ^{2}}}}
\label{eq:concentrationDef}
\end{equation}
where $\hat{U}$ and $\hat{V}$ are the unitary operators corresponding
to ''periodic observables'' of position on the torus. In terms of
the usual position observables $X$ and $Y$, where we have 
\[
\frac{-L+1}{2}\leq X,Y\leq\frac{L-1}{2},
\]
we can define these commuting unitary matrices as
\[
\hat{U}=e^{\frac{2\pi i}{L}X},\quad\hat{V}=e^{\frac{2\pi i}{L}Y}.
\]
 The quantity in Equation~\ref{eq:concentrationDef} is expected
to be similar to 
\begin{equation}
\mbox{localGap\ensuremath{\left(F\right)}}
=
\left\Vert \left(B(\eta X,\eta Y,H-E_{F})\right)^{-1}\right\Vert ^{-1}
\label{eq:GapAtOrigin}
\end{equation}
in the new method. It is slow to compute because computing the Fermi
projector is slow. 

By the Fermi projector we mean the spectral subspace of $H_{\mathrm{per}}$
corresponding to $(-\infty,E_{F}]$. In the event of a large spectral
gap we can prove that $P$ will have relatively small commutator with
$\hat{U}$ and $\hat{V}$, as the indicator function can be calculated
as $f(H_{\mathrm{per}})$ for $f$ with reasonable Fourier transform.
There is also what is called a mobility gap \cite{ZhangMobiltyGap},
where there is a region of the spectrum around the Fermi level that
is filled with eigenvalues that have well localized eigenstates. In
that case as well, the commutators $[P,\hat{U}]$ and $[P,\hat{V}]$
tend to be small.

There is an integer we can calculate here, the Bott index. Let $U=P\hat{U}P+(I-P)$
and $V=P\hat{V}P+(I-P)$ and define the index 
\[
\Re\left(\mathrm{Trace}\left(\frac{1}{2\pi i}\log\left(VUV^{*}U^{*}\right)\right)\right)
\]
which can be proven to be an integer. This integer will be zero when
$U$, $V$ and $P$ are close to a commuting triple of matrices and
when it is nonzero such as approximation is precluded. Given a full
eigensolve of $H_{\mathrm{per}}$, if assemble all the low-energy
eigenstates to form a non-square matrix $W$ with $WW^{*}=P$, then
a reformulation of this formula is
\[
\Re\left(\mathrm{Trace}\left(\frac{1}{2\pi i}\log\left(W^{*}\hat{V}P\hat{U}P\hat{V}^{*}P\hat{U}^{*}W\right)\right)\right).
\]
We can compute this from just the eigenvalues of $W^{*}\hat{V}P\hat{U}P\hat{V}^{*}P\hat{U}^{*}W$.
We do need up to half of the eigenvectors of $H_{\mathrm{per}}$,
so an algorithm for the Bott index is easy to implement in $O(n^{3})$
time, but no better.

The results using the old method are shown in Figure~\ref{fig:Chern-average-old}.
Both methods show a sharpening transition, with the new algorithm able
to work with large lattices.

\subsection{Class AII in 2D}

Now we look at the model used in the the numerical study done with
Hastings \cite{LorHastHgTe}, which was the model for an HgTe quantum
well given in \cite{konig2008quantum}. We keep the same disorder and the
same strength of the $H_\mathrm{BIA}$ term that breaks in version symmetry as
in the old study \cite{LorHastHgTe}.

We are just claiming proof of concept, the our formula in a possible
replacement for the Pfaffian-Bott index. We can't get to much larger
matrices than before because our algorithm to compute the sign of
the Pfaffian uses dense matrices. We hope this data will inspire the
production new software implementing the sparse matrix factorization
in \cite{DuffSparseFactorSkew}.

We look again at the local index and pseudospectrum, with $\eta=0.5.$
Figures \ref{fig:AII-local-1}-\ref{fig:AII-local-4} show this with
increasing disorder. The red $\times$ indicate index $-1$ while
the blue $\circ$ indicate index $1$.

\begin{figure}[pt]
\hspace*{-2in}(a)\includegraphics[scale=0.4]{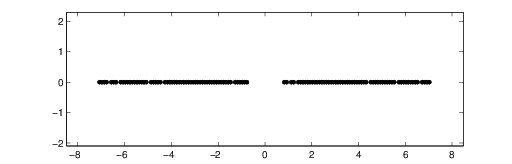}\\
\hspace*{-2in}(b)\includegraphics[scale=0.4]{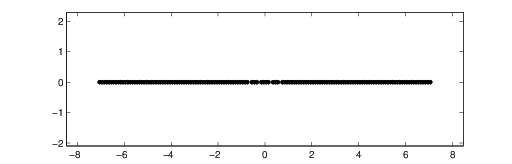}

\vspace*{-2in}

\hspace*{3in}(c)\includegraphics[clip,scale=0.55]{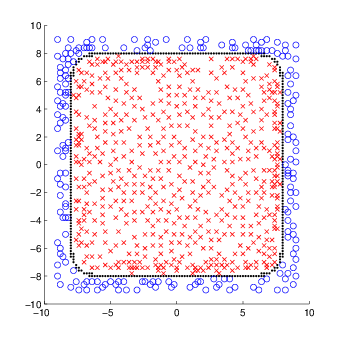}

\caption{Class AII insulator with disorder at 0. 
Panel (a) is the $0.1$-pseudospectrum of the Hamiltonian with periodic
boundary conditions. Panel (b) is the $0.1$-pseudo-spectrum of the
Hamiltonian with zero boundary conditions. Panel (c) is the Clifford
pseudospecturm with $K$-theory labels in the gaps.\label{fig:AII-local-1}}
\end{figure}

\begin{figure}[pt]
\hspace*{-2in}(a)\includegraphics[scale=0.4]{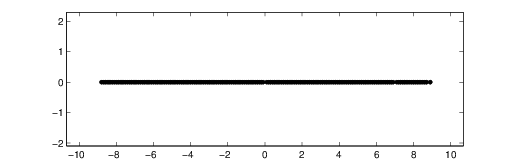}\\
\hspace*{-2in}(b)\includegraphics[scale=0.4]{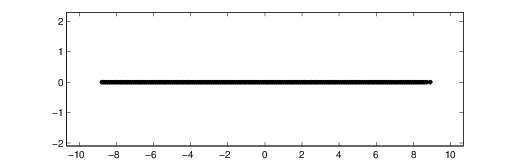}

\vspace*{-2in}

\hspace*{3in}(c)\includegraphics[clip,scale=0.55]{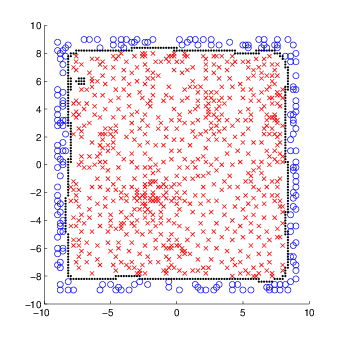}

\caption{Class AII insulator with disorder set to $7$. 
Panel (a) is the $0.1$-pseudospectrum of the Hamiltonian with periodic
boundary conditions. Panel (b) is the $0.1$-pseudo-spectrum of the
Hamiltonian with zero boundary conditions. Panel (c) is the Clifford
pseudospectrum. At this strength of disorder, the bulk gap is just
closing. 
\label{fig:AII-local-2}
}
\end{figure}

\begin{figure}
\includegraphics[clip,scale=0.55]{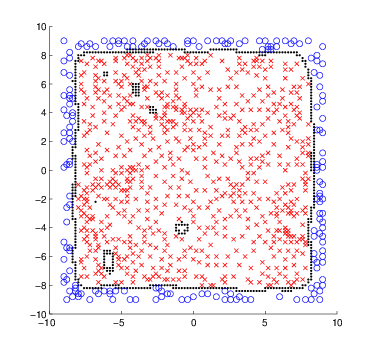}
\includegraphics[clip,scale=0.55]{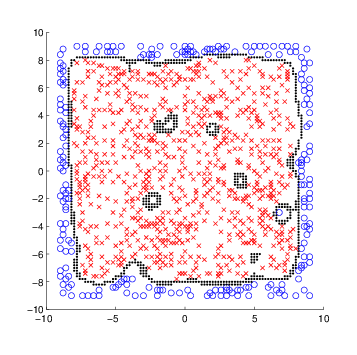}

\caption{Class AII insulator with disorder at 8 (left) and 9 (right).
\label{fig:AII-local-3}
}
\end{figure}

\begin{figure}
\includegraphics[clip,scale=0.55]{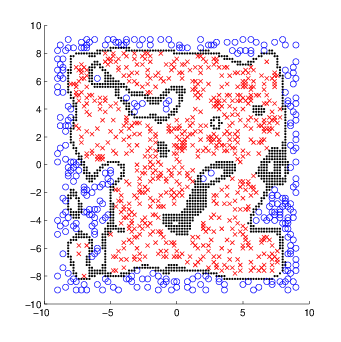}
\includegraphics[clip,scale=0.55]{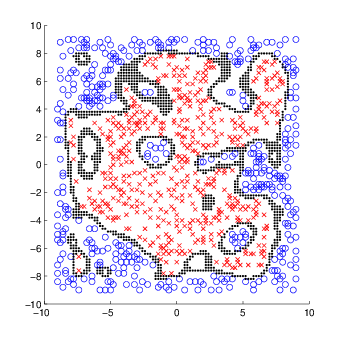}

\caption{Class AII insulator with disorder at 10 (left) and 11 (right).
\label{fig:AII-local-4}
}
\end{figure}

Now the global index. Here we set $\eta=4/L$ for an $L$-by-$L$
lattice. Figure \ref{fig:Phase-transition-2D-spin-old} is reproduced
from \cite{LorHastHgTe}. We cannot work with larger matrices in this
symmetry class because we do not have software to compute the needed
factorization of sparse antisymmetric matrices. The results of the
new formulas are shown in Figure \ref{fig:Phase-transition-2D-spin-new}.
We see that the transition is probably not sharpening as system size
increases, but it is hard to tell without the larger system sizes.
Notice that $40$ lattice units is roughly two nanometers. If we are
modeling films of roughly one nanometer thickness, we ought to be
looking at $L\approx200$. Such a size may be in reach of the sparse
algorithm as soon as that is available.

\begin{figure}
\includegraphics[scale=0.28]{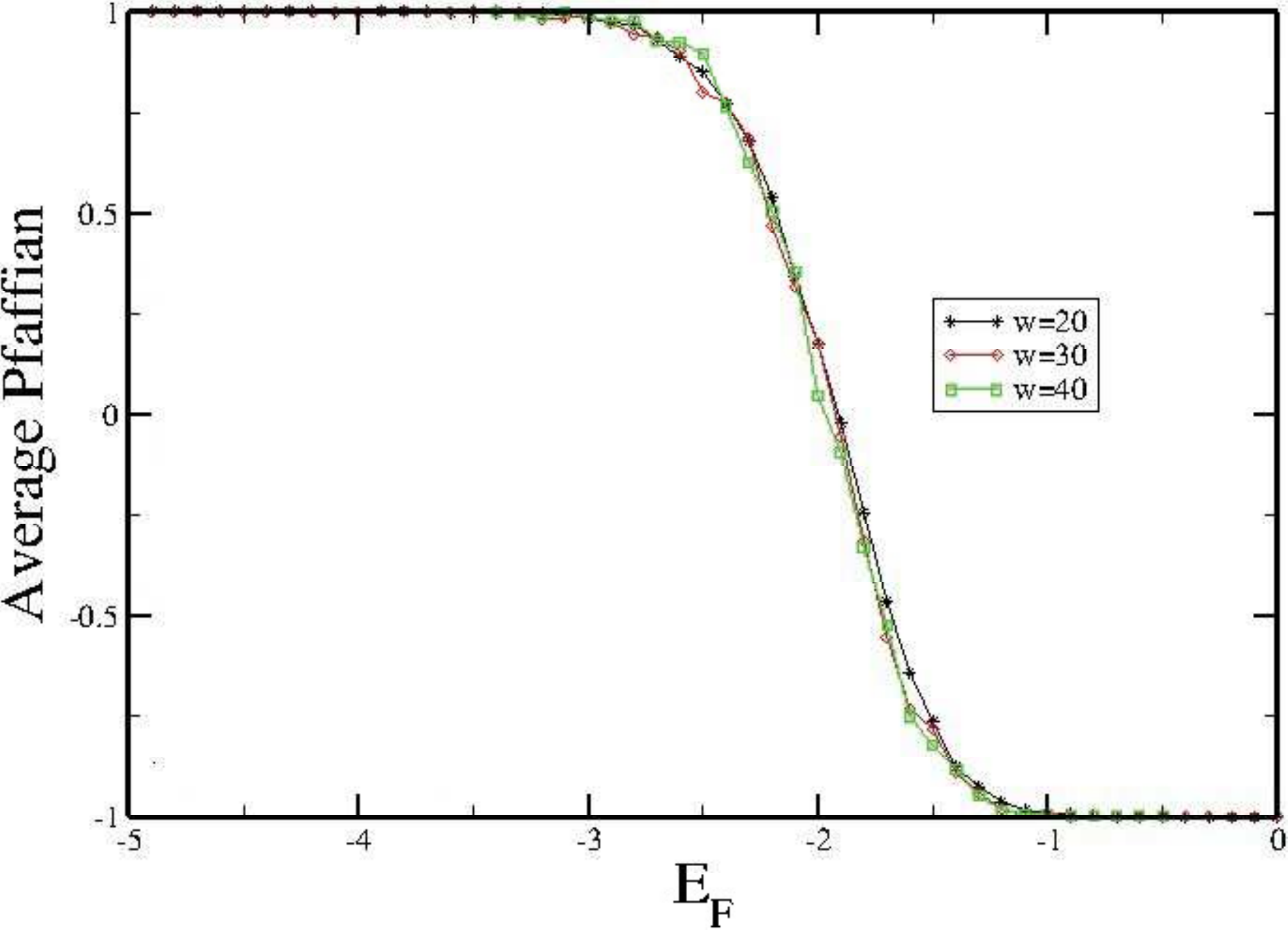}\\
\vspace*{-0.2cm}

\caption{Phase transition in 2D. The larger plot shows the disorder averaged
Pfaffian-Bott index. Reproduced from \cite{LorHastHgTe}. 
\label{fig:Phase-transition-2D-spin-old}
}
\end{figure}

\begin{figure}
\includegraphics[bb=1in 3in 7.5in 8in,clip,scale=0.45]{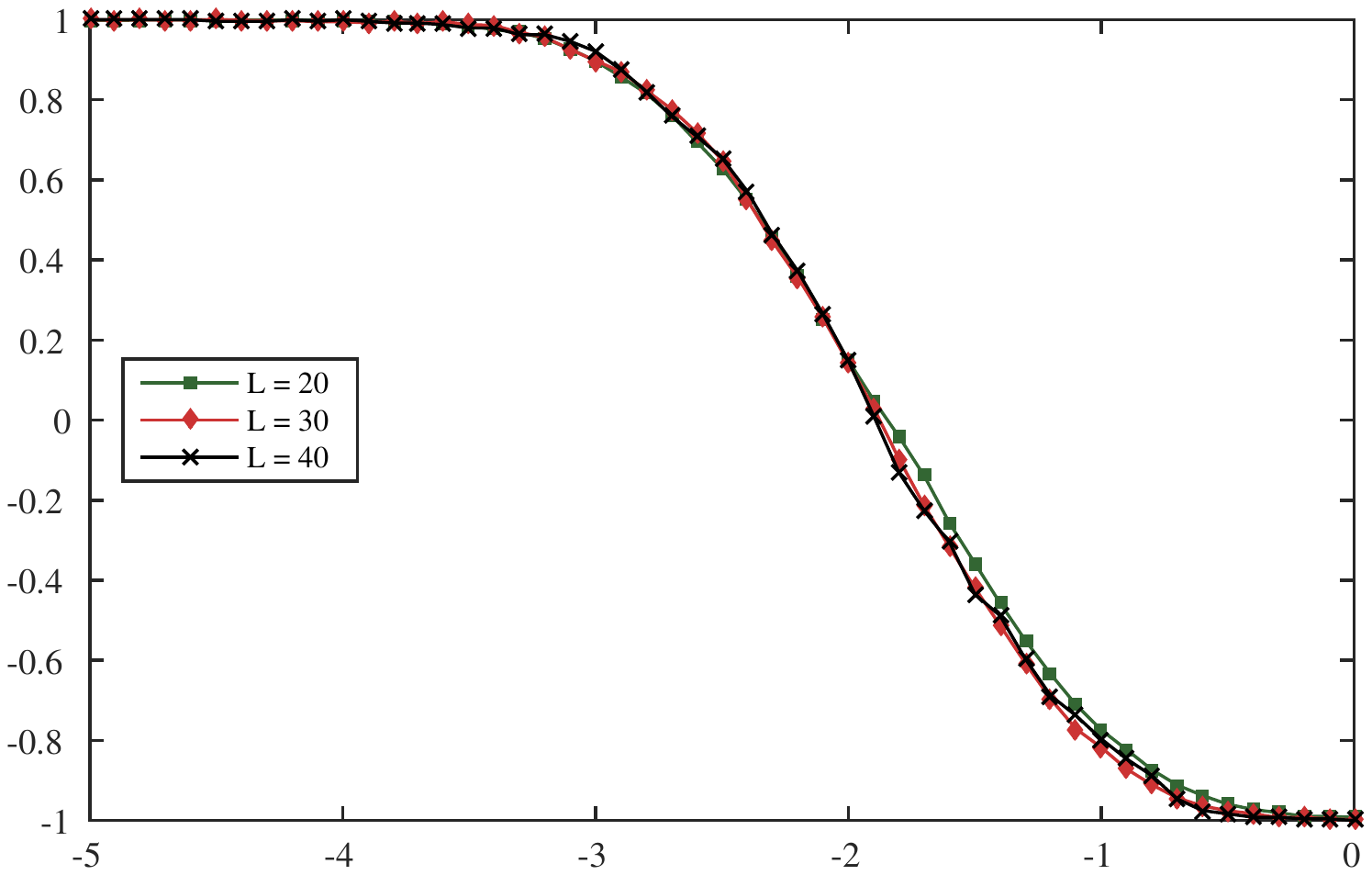}
\includegraphics[bb=1in 3in 7.5in 8in,clip,scale=0.45]{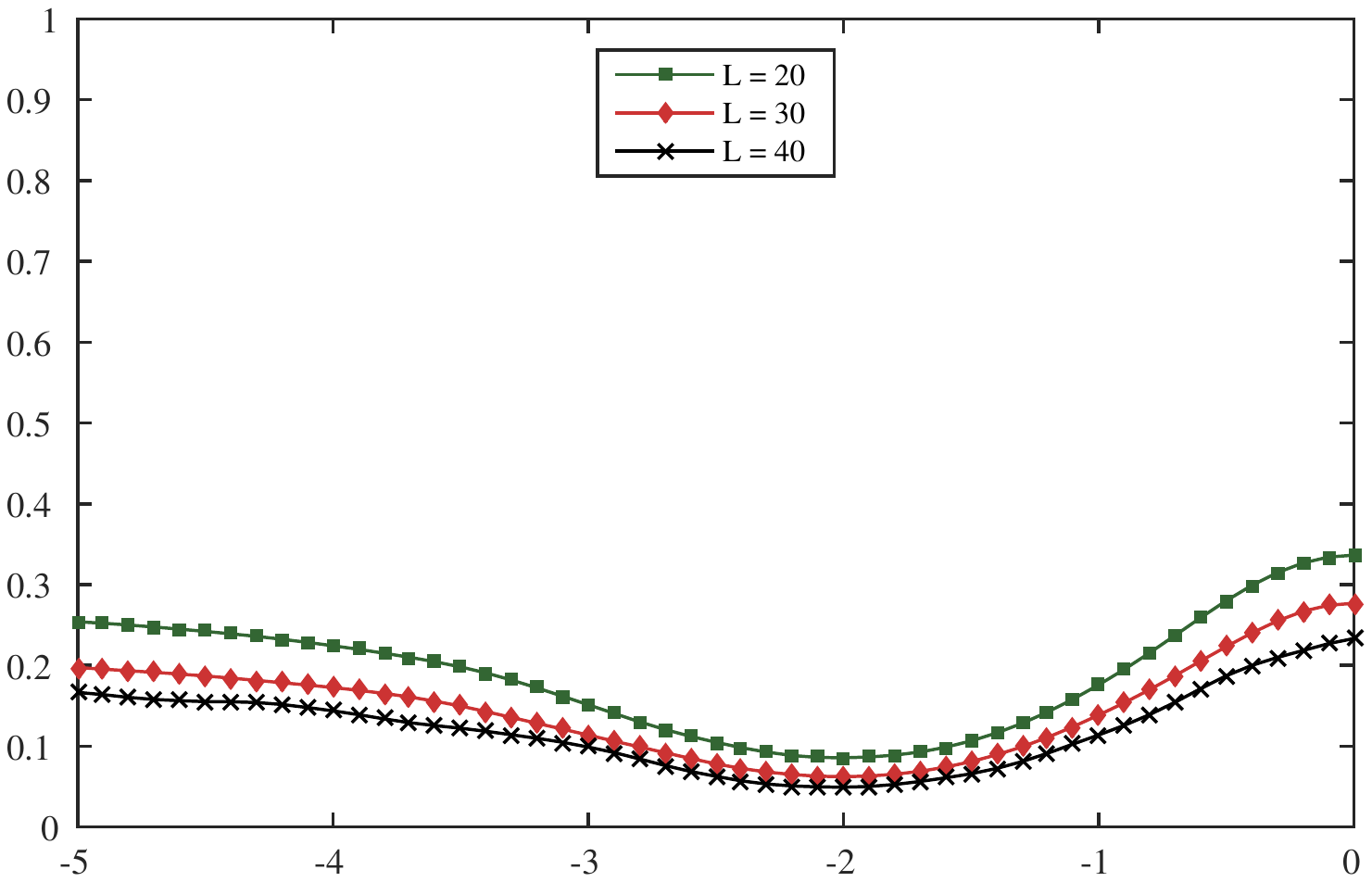}
\vspace*{-0.2cm}

\caption{Phase transition in 2D. The left panel shows the disorder averaged
new class AII index. The right panel shows the disorder-average of
gap localized at the origin. The number $N$ of samples for an $L$-by-$L$
lattice used in these averages was: $L=20$, $N=3544$; $L=30$, $N=2651$;
$L=40$, $N=481$. 
\label{fig:Phase-transition-2D-spin-new}
}
\end{figure}

\section{Dimension three, numerics}

We present examples in one symmetry class, AII.

\subsection{Class AII in 3D}

The first numerical study \cite{HastLorTheoryPractice} in 3D of the
effect of doping a topological insulator used an index  that worked for
periodic boundary conditions.  It also involved calculating the
sign of a determinant, but involving polynomials in three variables
that approximate a degree-one mapping of a three-torus to a three-sphere.
The geometry here forces these polynomials to have degree eleven,
and the resulting algorithm was slow. In contrast, as we are working
with a cube not a three-torus, we have a substantially faster algorithm. 

The local index was defined using $\eta=0.25$. The results are shown
in Figures \ref{fig:3D-insulator-local-1}-\ref{fig:3D-insulator-local-4}.
The red $\times$ indicate index $-1$ while the blue $\circ$ indicate
index $1$.

\begin{figure}
\hspace*{-2in}(a)\includegraphics[scale=0.4]{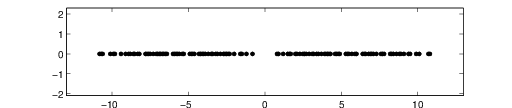}\\
\hspace*{-2in}(b)\includegraphics[scale=0.4]{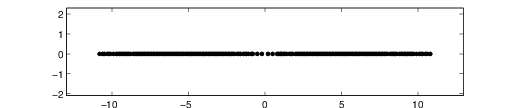}

\vspace*{-1.4in}

\hspace*{3in}(c)\includegraphics[clip,scale=0.55]{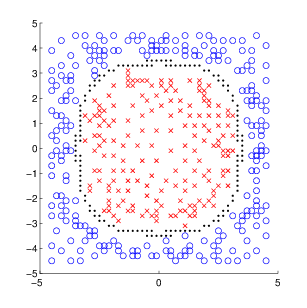}

\caption{Class AII insulator in 3D with no disorder. Showing a slice at $z=0$
and at the Fermi level. 
Using a 9-by-9-by-9 lattice. \label{fig:3D-insulator-local-1}}
\end{figure}
\begin{figure}
\hspace*{-2in}(a)\includegraphics[scale=0.4]{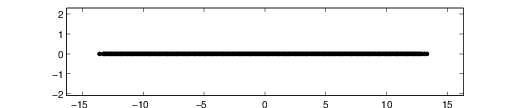}\\
\hspace*{-2in}(b)\includegraphics[scale=0.4]{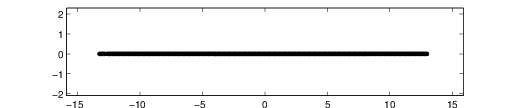}

\vspace*{-1.4in}

\hspace*{3in}(c)\includegraphics[clip,scale=0.55]{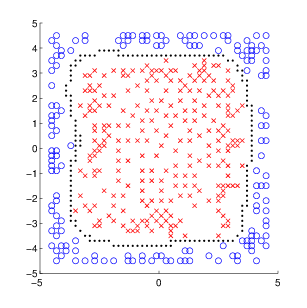}

\caption{Class AII insulator in 3D with disorder at 11. Showing a slice at
$z=0$ and at the Fermi level. Using a 9-by-9-by-9 lattice. 
The bulk gap is just closing at this disorder level. 
\label{fig:3D-insulator-local-2}
}
\end{figure}

\begin{figure}
\includegraphics[clip,scale=0.55]{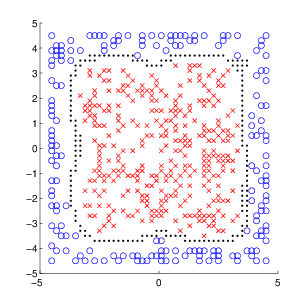}
\includegraphics[clip,scale=0.55]{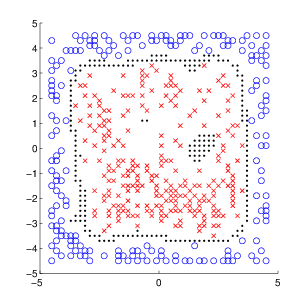}

\caption{Class AII insulator in 3D with disorder at 12 (left) and 13 (right).
\label{fig:3D-insulator-local-3}
}
\end{figure}
\begin{figure}
\includegraphics[clip,scale=0.55]{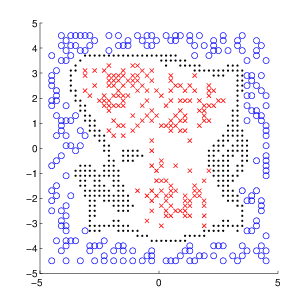}
\includegraphics[clip,scale=0.55]{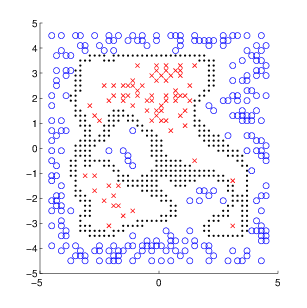}

\caption{Class AII insulator in 3D with disorder at 14 (left) and 15 (right).
\label{fig:3D-insulator-local-4}
}
\end{figure}

The global index was defined using $\eta=4/L$ for an $L$-by-$L$-by-$L$
lattice. Figure~7.2 of \cite{HastLorTheoryPractice} is replicated
here as Figure~\ref{fig:Phase-transition-3D-old}. This data was
sufficiently noisy that no real conclusions about scaling could be
made. The left panel in Figure~\ref{fig:Phase-transition-3D-new}
shows how we can generate much cleaner data with the new algorithm
by using more samples and larger systems. It appears now the transition
from this 3D topological insulator to an ordinary insulator is not
sharp, but larger system sizes are needed to clarify this. These can
be studied on existing machines, but the processing time needed will be
significant.

\begin{figure}
\includegraphics[bb=2bp 0bp 712bp 530bp,clip,scale=0.28]{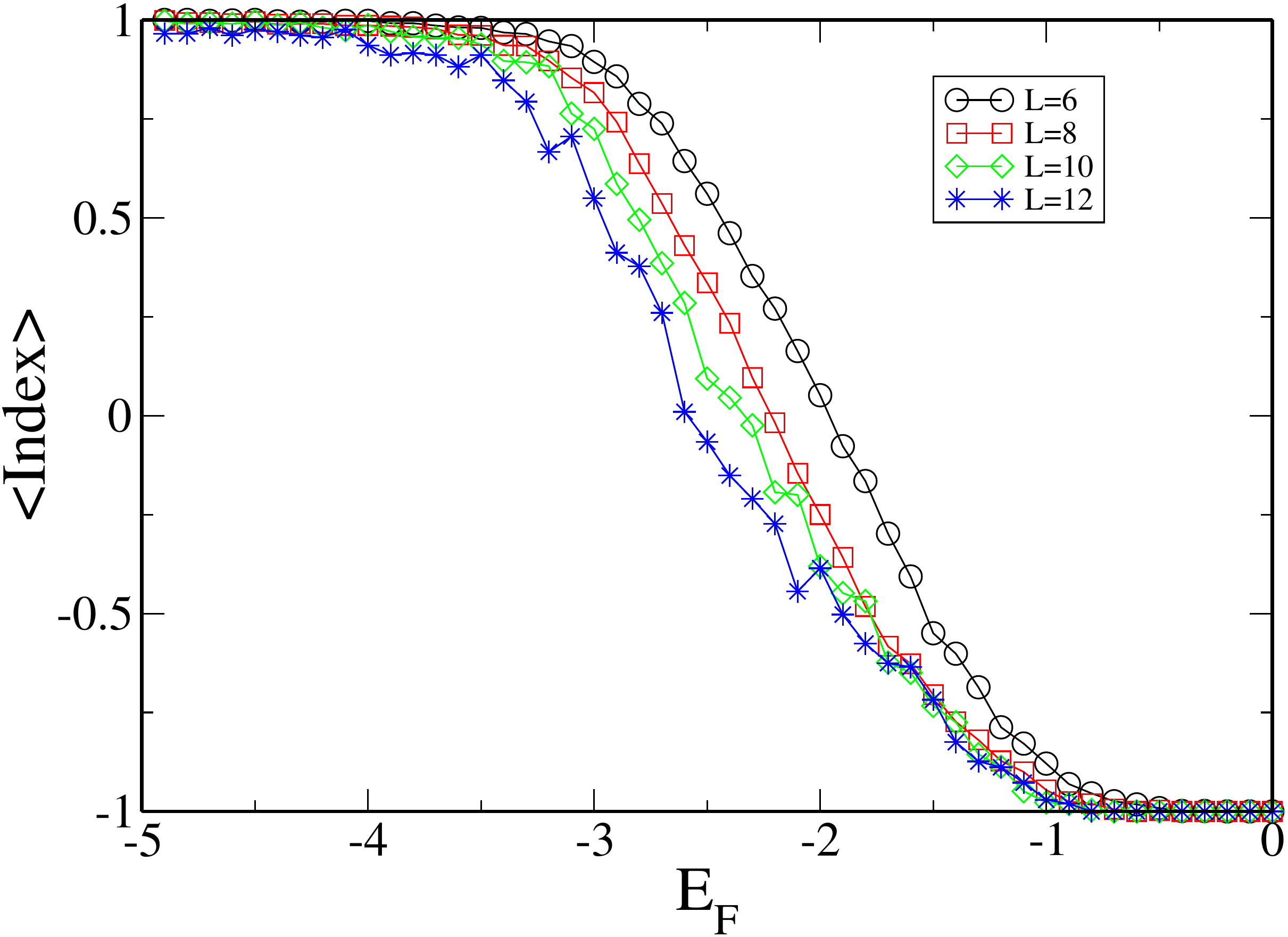}\vspace*{-0.2cm}

\caption{Phase transition in 3D, as computed via the old method.  
This figure is replicated from \cite{HastLorTheoryPractice}.
Shown is a plot of average index for
L = 6, 8, 10, 12 with $L\times L\times L$ lattices. Each data point
is an average of 1700, 1400, 600, 400 samples, respectively. 
\label{fig:Phase-transition-3D-old}
}
\end{figure}

\begin{figure}
\includegraphics[bb=1in 3in 7.5in 8in,clip,scale=0.45]{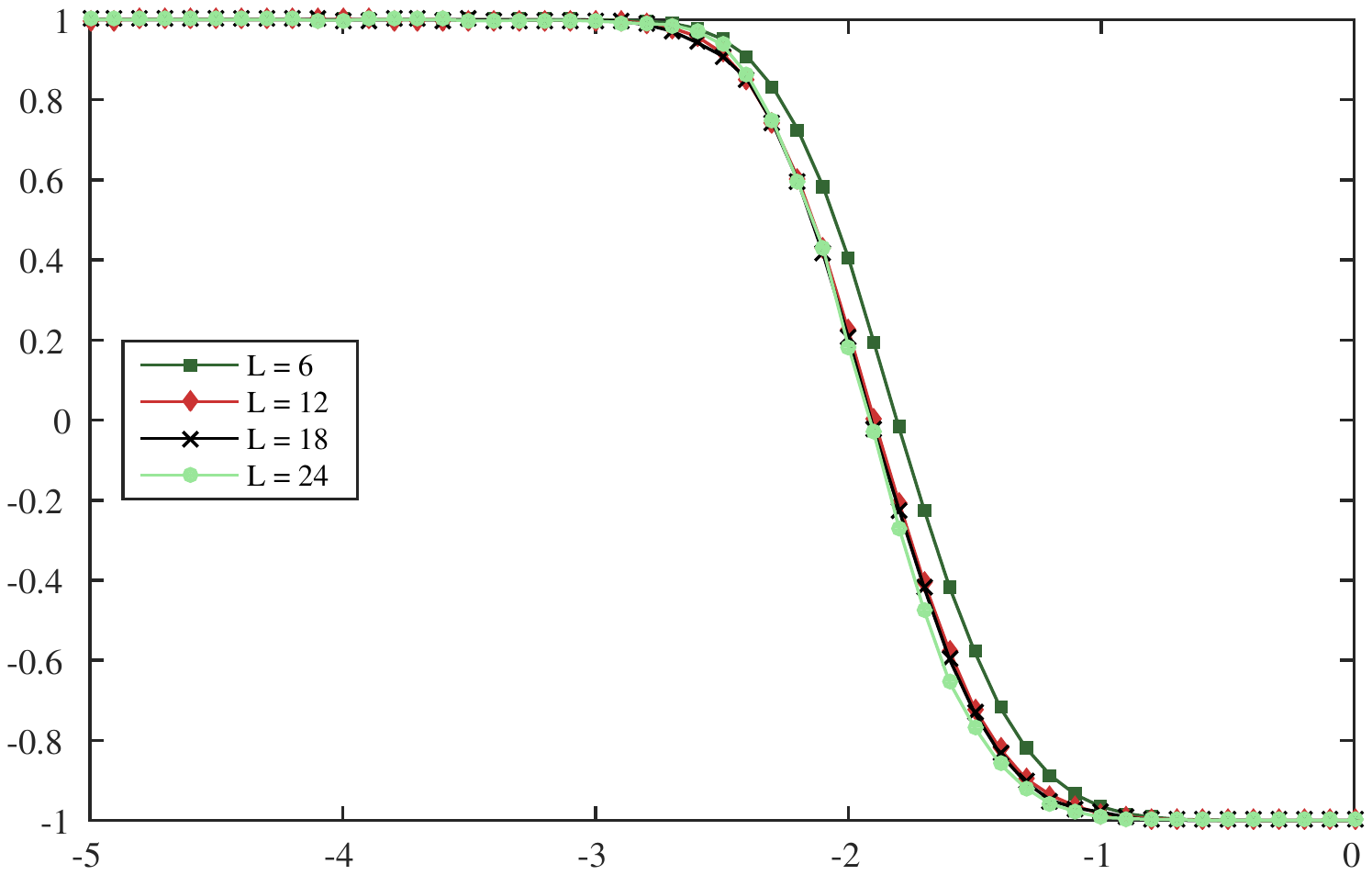}
\includegraphics[bb=1in 3in 7.5in 8in,clip,scale=0.45]{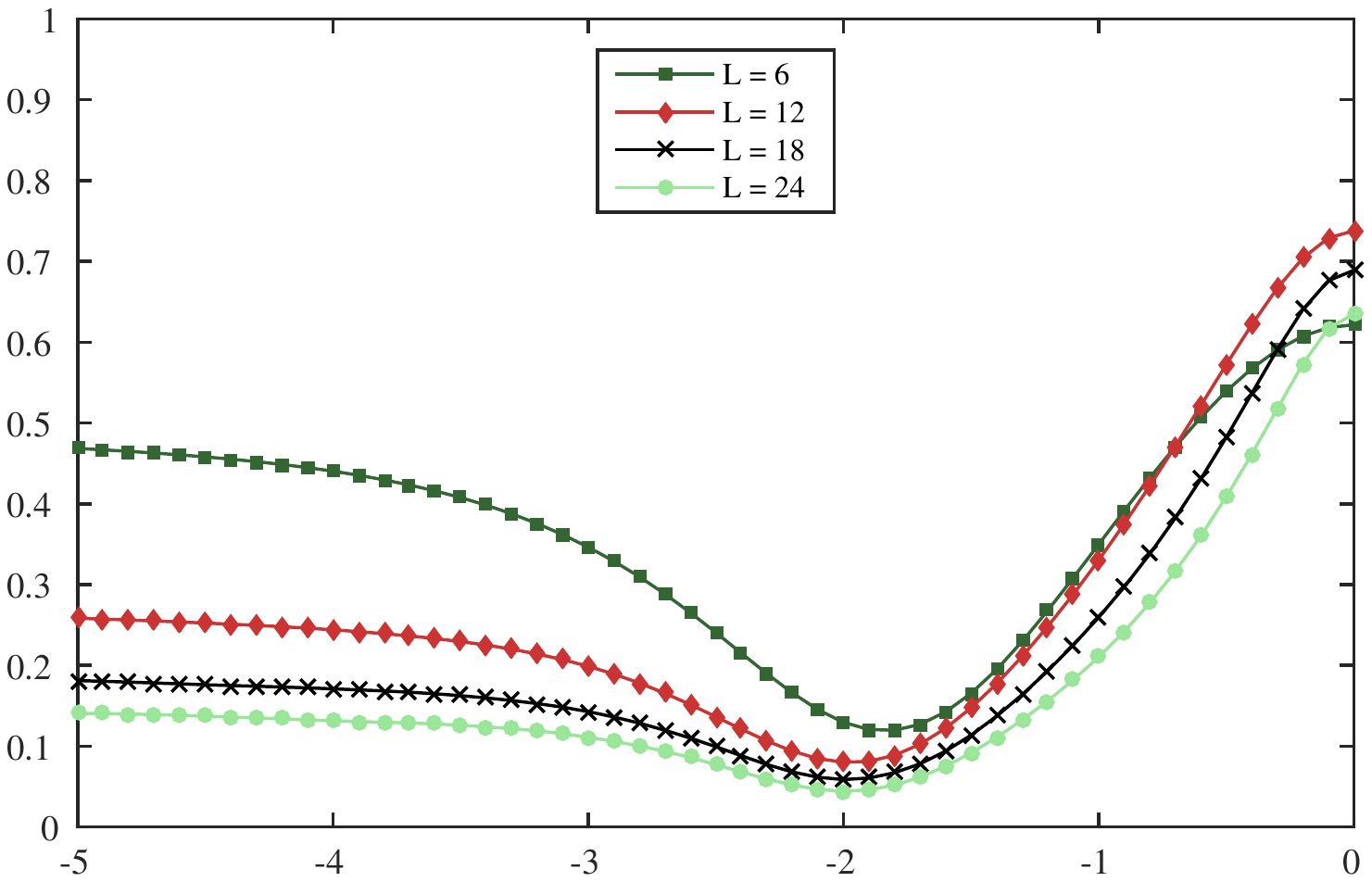}\vspace*{-0.2cm}

\caption{Phase transition in 3D. The left panel shows the disorder averaged
new class AII index. The right panel shows the disorder-average of
gap localized at the origin. The number $N$ of samples for an $L$-by-$L$-by-$L$
lattice used in these averages was: $L=6$, $N=26036$; $L=12$, $N=23357$;
$L=18$, $N=3491$; $L=24$, $N=1547$. 
\label{fig:Phase-transition-3D-new}
}
\end{figure}

\section{The algorithms}

Matlab code, with instructions on how to produce many of the figures
in this paper, will be made
available at a data repository\footnote{ 
https://repository.unm.edu/handle/1928/23449}.
For the larger system sizes, the study of the global index was done
on multiple computing nodes, each with 32 cores and 64GB of random
access memory. However the local index, at the system sizes illustrated
in the figures, can be explored using less than a day on a desktop
with 4 cores and only 8GB of random access memory.

A good example here is the algorithm for the global class AII invariant
in 3D. The formula in equation \ref{eq:3D-AII-full formula} tells
us we need to first compute a matrix
\[
A=Q^{*}
\left[\begin{array}{cccc}
0 & 0 & H+iZ & iX+Y\\
0 & 0 & iX-Y & H-iZ\\
H-iZ & -iX-Y & 0 & 0\\
-iX+Y & H+iZ & 0 & 0
\end{array}\right]
Q
\]
 that is real and sparse. We need to compute its spectral gap
\[
\left\Vert A^{-1}\right\Vert ^{-1}
\]
as well as the sign of its determinant.

We rely on LU algorithm \cite{DavisSparseLU}, as implemented in Matlab,
to factor $A$ as
\begin{equation}
A=R^{-1}P^{*}LUQ^{*}\label{eq:LURPQ}
\end{equation}
where $R$ is diagonal, $P$ and $Q$ are permutation matrices, $L$
is lower triangular, sparse with unit diagonal, and $U$ is upper
triangular and sparse. The determinant of $L$ is one, and so
\[
\mathrm{sign}(\det(A))=\mathrm{sign}(\det(R))\mathrm{sign}(\det(P))\mathrm{sign}(\det(L))\mathrm{sign}(\det(Q)).
\]
These signs of each these determinants is easy to compute. The norm of $A^{-1}$
we compute with the power method. Essentially this method starts with
a random unit vector $\mathbf{v}=\mathbf{v}_{0}$ and then computes
a few dozen iterations of 
\begin{align*}
\mathbf{w}_{n} & =\left(A^{*}A\right)^{-1}\mathbf{v}_{n-1}.\\
\mathbf{v}_{n} & =\frac{1}{\left\Vert \mathbf{w}_{n}\right\Vert }\mathbf{w}_{n}.
\end{align*}
We modified this procedure a little. We found starting with $\mathbf{v}_{0}$
having all entries equal worked better than a random vector in this
setting. Then, following \cite{TrefethenEmbree}, we compute $\mathbf{w}_{n}$
using equation \ref{eq:LURPQ} and the Matlab $\setminus$ operator
that computes $C^{-1}\mathbf{x}$ without inverting the matrix. Since
\[
\left(A^{*}A\right)^{-1}=QU^{*}L^{*}PR^{-2}P^{*}LUQ^{*},
\]
and since $\left(A^{*}A\right)^{-1}$ and $Q^{*}\left(A^{*}A\right)^{-1}Q,$
we can compute $\mathbf{w}_{n}$ via
\[
\mathbf{w}_{n}
=
U^{*}\setminus(L^{*}\setminus(P\setminus(R^{-2}\setminus(P^{*}\setminus(L\setminus(U\setminus\mathbf{v}_{n-1})))))).
\]

\section*{Acknowledgments}

The author wish to thank Deborah Evans, Alexei Kitaev, Matthew Hastings,
Joel Moore and Hermann Schulz-Baldes for illuminating discussions, 
mathematical and physical.  

This work was partially supported by a grant from the Simons Foundation
(208723 to Loring) and by financial support form the Erwin Schr\"odinger
International Institute for Mathematical Physics.  Most of the comupting was
done on machines at the
Center for Advance Research Computing at the University of New Mexico.

\rule[0.5ex]{1\linewidth}{1pt}

\bibliographystyle{plain}
\bibliography{/home/terry/opAlgResearch/cstarRefs}

\end{document}